\documentclass[Journal]{IEEEtran}

\usepackage{amsmath,amssymb,float,arydshln,color}
\usepackage{amsthm}
\usepackage{cite}
\usepackage{psfrag,setspace,wrapfig,subfigure}
\usepackage[latin1]{inputenc}
\usepackage{dsfont}
\usepackage{epsfig}
\usepackage{epstopdf}
\usepackage{graphicx}
\usepackage{hyperref}
\usepackage{amsfonts}
\usepackage{bbm}
\allowdisplaybreaks

\newtheorem{lemma}{{Lemma}}
\newtheorem*{assumption*}{{ Assumption}}
\newtheorem{theorem}{{Theorem}}

\def\tran{^{\mathsf{T}}}
\def\one{\mathds{1}}

\newcommand{\bp}{\small \begin{proof}}
	\newcommand{\ep}{\end{proof} \normalsize}

\newcommand{\Ex}{\mathbb{E}\hspace{0.05cm}}
\newcommand{\Exf}{\mathbb{E}_{\mathcal{F}_{i}}\hspace{0.05cm}}

\newcommand{\bm}[1]{\mbox{\boldmath $#1$}}

\newcommand{\be}{\begin{equation}}
\newcommand{\ee}{\end{equation}}
\newcommand{\bal}{\begin{align}}
\newcommand{\eal}{\end{align}}
\newcommand{\bq}{\begin{eqnarray}}
\newcommand{\eq}{\end{eqnarray}}
\newcommand{\bqn}{\begin{eqnarray*}}
	\newcommand{\eqn}{\end{eqnarray*}}
\newcommand{\nn}{\nonumber}
\newcommand{\ba}{\left[ \begin{array}}
	\newcommand{\ea}{\\ \end{array} \right]}
\newcommand{\qd}{\hfill{$\blacksquare$}}
\newcommand{\define}{\;\stackrel{\Delta}{=}\;}
\newcommand{\aseq}{\;\stackrel{a.s.}{=}\;}

\def\bxi  	{{\boldsymbol \xi}}
\def\bmu  	{{\boldsymbol \mu}}
\def\bpsi	{{\boldsymbol \psi}}

\def\Z{{\boldsymbol{Z}}}

\def\m{{\boldsymbol{m}}}

\def\Zint{{\mathchoice{\setbox1=\hbox{\sf Z}\copy1\kern-.75\wd1\box1}
		{\setbox1=\hbox{\sf Z}\copy1\kern-.75\wd1\box1}
		{\setbox1=\hbox{\scriptsize\sf Z}\copy1\kern-.75\wd1\box1}
		{\setbox1=\hbox{\scriptsize\sf Z}\copy1\kern-.75\wd1\box1}}}

\makeatletter
\def\hlinewd#1{%
	\noalign{\ifnum0=`}\fi\hrule \@height #1 \futurelet
	\reserved@a\@xhline}
\makeatother

\begin{document}
	\def\helvetica{phvr7t.tfm}
	\def\helveticaoblique{phvro7t.tfm}
	\def\helveticabold{phvb7t.tfm}
	\def\helveticaboldoblique{phvbo7t.tfm}
	
	\font\sfb=\helveticabold
	=\helveticaboldoblique
	\title{Social Learning over Weakly-Connected Graphs}
	
	\author{Hawraa~Salami,~\IEEEmembership{Student Member,~IEEE,}
		Bicheng~Ying,~\IEEEmembership{Student Member,~IEEE,}\\
		and~Ali~H.~Sayed,~\IEEEmembership{Fellow,~IEEE}
	
		\thanks{This work was supported in part by NSF grants CCF-1524250 and ECCS-1407712, DARPA project N66001--14--2-4029, and by a Visiting Professorship from the Leverhulme Trust, United Kingdom. An early short version of this work appears in the conference publication \cite{icassp2}.}
		\thanks{The authors are with Department of Electrical Engineering, University of California, Los Angeles, CA 90025. Emails: \{hsalami, ybc, sayed\}@ucla.edu}
	}
	\maketitle
	
	\begin{abstract}
		In this paper, we study diffusion social learning over weakly-connected graphs. We show that the asymmetric flow of information hinders the learning abilities of certain agents regardless of their local observations. Under some circumstances that we clarify in this work, a scenario of {\em{total influence}} (or ``mind-control'') arises where a set of influential agents ends up shaping the beliefs of non-influential agents. We derive useful closed-form expressions that characterize this influence, and which can be used to motivate design problems to control it. We provide simulation examples to illustrate the results.
	\end{abstract}
	
	
	\section{Introduction and Related Works}
	Social interactions among networked agents influence the beliefs of agents about
	the state of nature. For example, in deciding whether the
	state of nature, denoted by $\bm{\theta}$, is either $\bm{\theta}=1$ or $\bm{\theta}=0$, an agent $k$
	observes some data whose probability distribution is dependent on the unknown $\bm{\theta}$ and,
	additionally, consults with its neighboring agents about their opinion on the most plausible value for $
	\bm{\theta}$. By combining their local measurements with the information from their neighbors, agents
	update their belief about $\bm{\theta}$ continuously.
	
	There are two main categories of models that have
	been proposed to examine this evolving interaction process \cite{chamley2004rational},\cite{OpinionDynamic}. In the first (Bayesian learning) category, the agents rely on some priors and on Bayes' rule to update their beliefs \cite{chamley2004rational},\cite{smith2000pathological},\cite{5717483},\cite{AcemogluDah},\cite{krishna1},\cite{krishna2}. In the second (non-Bayesian learning) category, agents interact with their neighbors and aggregate their beliefs into their own \cite{degroot,misinformation,epstein2010non,jackson,jadbabaie2012non,molavi2013reaching,zhao2012learning}. One notable example of non-Bayesian learning is \cite{jadbabaie2012non} where the authors proposed a consensus-type construction to update the agents' beliefs. In this construction, each agent follows the Bayes' rule to obtain an intermediate belief and subsequently combines it with the {\em old} beliefs of its neighbors. Under some technical assumptions, it was shown in \cite{jadbabaie2012non} that agents following this model can asymptotically learn the true state.
	
	Motivated by this study, an alternative to the consensus mechanism was proposed in \cite{zhao2012learning} by relying on diffusion strategies due to their enhanced performance and stability ranges, especially in scenarios that involve continuous learning \cite{Sayed,sayed2014adaptive}. In the diffusion-based model, each agent combines its intermediate belief with the \textit{updated} (rather than old) beliefs of its neighbors. Results in \cite{zhao2012learning} established that agents are also able to asymptotically learn the underlying state under the diffusion strategy.
	
	The models of social interaction studied in \cite{jadbabaie2012non,zhao2012learning} assume {\em strongly-connected} graphs whereby a path with positive weights
	connecting any two agents is always possible and at least one agent has a self-loop. Over such graphs, social influences diffuse over time and all agents are able to learn asymptotically the true state of the environment. This is possible even when the local observations at the agents may be of varying quality with some agents being more informed than others. 
	
	\subsection{Weakly-Connected Networks}
	In this work, we examine social learning over {\em weakly-connected} graphs, as opposed to strongly-connected graphs. Over a weak topology, there exist some select edges over which information flows in one direction only, with information never flowing back from the receiving agents to the originating agents. This scenario is common in practice, especially over social networks. For example, in Twitter networks, it is not unusual for some influential agents (e.g., celebrities) to have a large number of followers, while the influential agent itself may not consult information from most of these followers. A similar effect arises when social networks operate in the presence of stubborn agents \cite{misinformation,mouraa,ying2014information};
	these agents insist on their opinion regardless of the evidence provided by local observations or by neighboring agents. It turns out that weak graphs influence the evolution of the agents' beliefs in a critical manner. The objective of this work is to clarify this effect, its origin, and to quantify its implications by means of closed-form expressions.
	
	\subsection{Social Disagreement}
	In the previous works \cite{ying2014information,icassp}, the authors examined the influence of weak graphs on the solution of distributed {\em inference} problems, where agents are interested in learning a parameter of interest that minimizes an aggregate cost function. It was shown there that a leader-follower relationship develops among the agents with the performance of some agents being fully controlled by the performance of other agents. In the different context of social learning, this type of weak connectivity was briefly discussed in \cite{molavi2013reaching} where consensus social learning was analyzed over non-strongly connected networks. This work considered only the special case in which all agents in the network are interested in the {\em same} state of nature. A richer and more revealing dynamics arises when different clusters within the network monitor different state variables. 
	
	For example, consider a situation in which a weak graph consists of four sub-graphs (see future Fig. \ref{fig.WCFigure}): the two top graphs are strongly-connected while the other two are weakly-connected to them. In this case, each of the first two sub-graphs is able to learn its truth asymptotically. However, the agents in the lower sub-graphs will be shown to reach a state of disarray in relation to their true state, with different agents reaching in general different conclusions and, moreover, with each of these conclusions being directly determined by the separate states of the two top sub-graphs. In this work we carry out a detailed analysis to show how influential agents dictate the performance of weak components in the network, and arrive at closed-form expressions that describe this influence in analytical form (suitable for subsequent design purposes). We will find that, under some conditions, non-influential agents will be forced to adopt beliefs centered around the true states of the influential agents. This situation is similar to the leader-follower relationship discussed in \cite{ying2014information,icassp} in the context of decentralized inference and continuous adaptation. We will also find that these beliefs differ from one agent to another, which results in a disturbing form of social disagreement. In some applications, the influential agents my be malicious as in \cite{malicious1,malicious2}. In contrast to these works, in our development, influential agents do not alter the information they are fusing, but the nature of what they are sending need not be consistent with the true state of the receiving agents.
	\subsection{Enhancing Self-Awareness}
	Motivated by the results in the next sections, we will also incorporate an element of self-awareness
	into the social learning process of the network through the introduction of a scaling factor --- see Eq. (\ref{model2}). This factor will enable agents in the network to assign more or less weight to their local information in comparison to the information received from their neighbors. This variation helps infuse into the network some elements of human behavior. For example, in an interactive social setting, a human agent may not be satisfied or convinced by an observation and prefers to give more weight to their prior belief based on accumulated experiences. This mode of operation was studied for {\em single} stand-alone agents in \cite{epstein2010non, epstein2} and was studied there as a mechanism for self-control. We will instead examine the influence of self-awareness in the challenging network setting, where the behavior of the various agents are coupled together. In particular, we will show that self-awareness helps agents converge towards a fixed belief distribution, rather than have their beliefs exhibit an undesired oscillatory behavior, which reflects their inability to settle on a decision --- see Fig. \ref{motiv}.
	
	\textit{Notation:} We use lowercase letters to denote vectors, uppercase letters for matrices, plain letters for deterministic variables, and boldface for random variables. We also use $(.)\tran$ for transposition, $(.)^{-1}$ for matrix inversion, and $\rho(.)$ for the spectral radius of a matrix. We use $\preceq$ and $\succeq$ for vector element-wise comparisons.
	
	\section{Strongly-Connected Networks}
	We first review strongly-connected networks and summarize the results already obtained over this graph topology. Then, we explain how the results are affected when the underlying topology happens to be weak and show how a leader-follower relationship develops. We characterize in some detail the limiting behavior of this relation and identify the factors that influence the ability of the social agents to learn the truth or to follow other influential agents. 
		\begin{figure}[h!]
			\centering
			\includegraphics[scale=0.35]{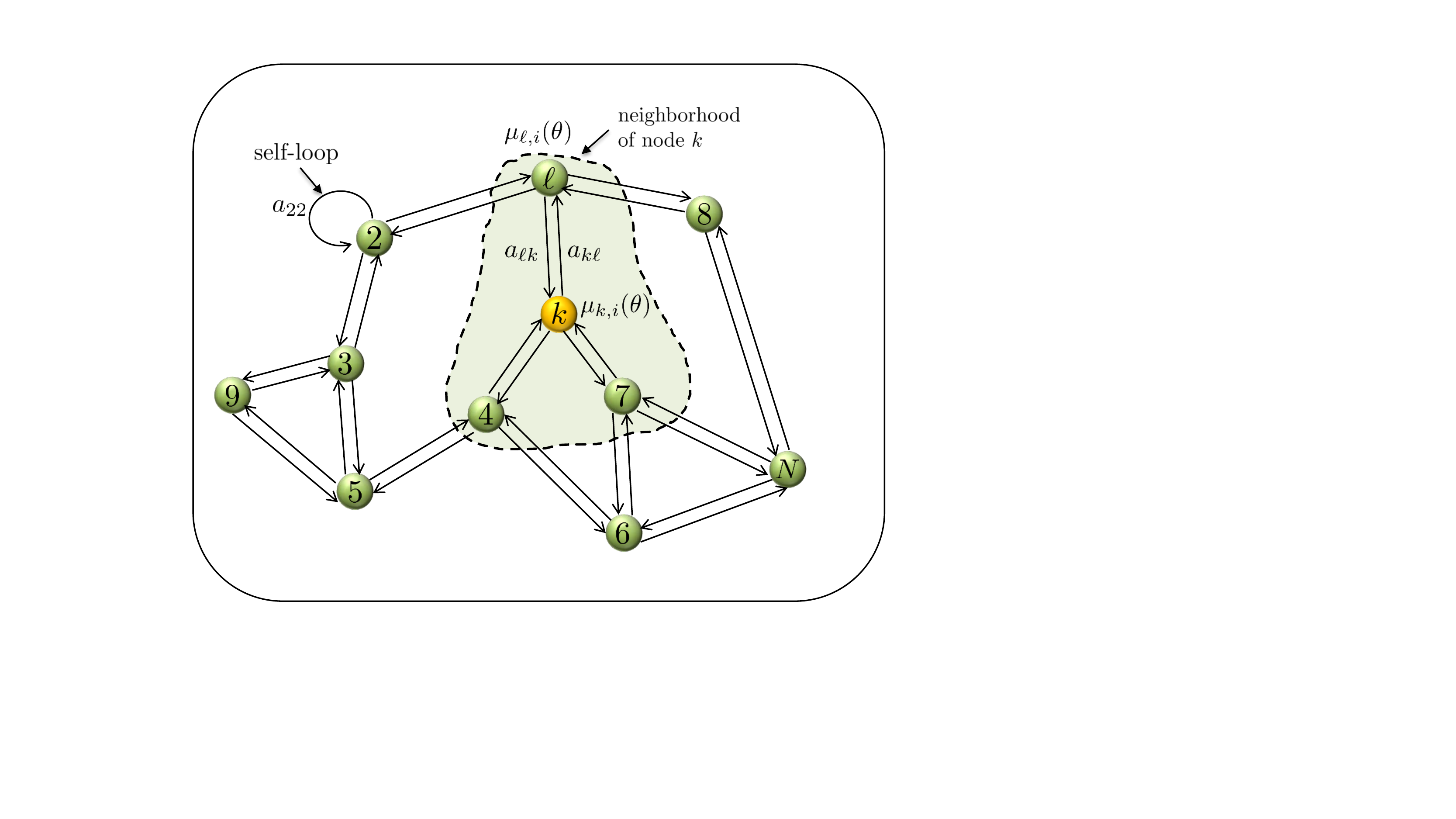} 
			\caption{\small An example of a strongly-connected network where $\mu_{k,i}(\theta)$ denotes the belief (pdf) of agent $k$ at time $i$.}
			\label{fig.SCFigure}
		\end{figure}
	\subsection{Network Model}
	Thus, consider a network of $N$ agents connected by some graph. Let $\mathcal{N}=\{1,2,\dots,N\}$ denote the indexes of the agents in the network. We assign a pair of non-negative weights, $\{a_{k\ell},a_{\ell k}\}$, to the edge connecting any two agents $k$ and $\ell$. The scalar $a_{\ell k}$ represents the weight with which agent $k$ scales the data arriving from agent $\ell$ and, similarly, for $a_{k\ell}$ -- see Fig. \ref{fig.SCFigure}. The network is said to be strongly-connected if there exists a path with non-zero weights connecting any two agents and, moreover, there is at least one self-loop, i.e., $a_{kk}>0$ for some agent $k$. Let $\mathcal{N}_k$ denote the neighborhood of agent $k$, which consists of all agents connected to $k$. Each agent $k$ scales data arriving from its neighbors in a convex manner, i.e.,
	\be
	a_{\ell k}\geq 0, \quad \sum_{\ell \in \mathcal{N}_k}a_{\ell k}=1, \quad a_{\ell k}=0 \text{ if } \ell\notin\mathcal{N}_k \label{convexCond}
	\ee
	We collect the weights $\{a_{\ell k}\}$ into an $N \times N$ matrix $A$. From condition (\ref{convexCond}), $A$ is a left-stochastic matrix so that its spectral radius is equal to one, $\rho(A)=1$. Since the network is strongly-connected, $A$ is also a primitive matrix \cite{Sayed}. It then follows from the Perron-Frobenius Theorem \cite{matrixAnalysis},\cite{perron} that $A$ has a single eigenvalue at one while all other eigenvalues are strictly inside the unit disc. We denote the right-eigenvector of $A$ that corresponds to the eigenvalue at one by $y$, and all entries of this vector will be strictly positive.  We normalize the entries of $y$ to add up to one, so that $y$ satisfies the following conditions:
	\be
	Ay=y,\quad \one\tran y=\one,\quad y \succ 0
	\label{perronFrob}
	\ee
	We refer to $y$ as the Perron eigenvector of $A$. This network structure plays an important role in diffusing information across the network and helps agents in learning the true state. We describe next the mechanism of this learning.
	
	\subsection{Diffusion Social Learning}
	Let $\Theta$ denote a finite set of all possible events that can be detected by the network. Let $\theta^\circ\in\Theta$ denote the \emph{unknown} true event that has happened, while the other elements in $\Theta$ represent possible variations of that event. The objective of the network is to learn the true state, $\theta^\circ$. For this purpose, agents will be continually updating their beliefs about the true state through a localized cooperative process. Initially, at time $i=0$, each agent $k$ starts from some prior belief, denoted by the function $\mu_{k,0}(\theta)\in[0,1]$. This function represents the probability distribution over the events $\theta\in\Theta$. For instance, if $\theta_1\in\Theta$ then
	\be \mu_{k,0}(\theta_1) = \mbox{\rm Prob}(\bm\theta = \theta_1),\;\;\;\;\;\mbox{\rm at time $i=0$}\ee
	For subsequent time instants $i\ge1$, the private belief of agent $k$ is denoted by $\mu_{k,i}(\theta)\in[0,1]$. All beliefs across all agents must be valid probability measures over $\Theta$. That is, they must obey the normalization:
	\be
	\sum_{\theta\in\Theta}\mu_{k,i}(\theta) = 1,\;\;\;\;\mbox{\rm for any $i\geq 0$ and $k \in \mathcal{N}$}
	\ee
	Figure \ref{fig.pmf} presents an example of a belief distribution $\mu_{k,i}(\theta)$ defined over $\Theta=\{\theta_1,\theta_2,\theta_3,\theta_4\}$. The agents will update their private beliefs $\{\mu_{k,i}(\theta)\}$ over time based on the private signals they observe from the environment and the information shared by their social neighbors. We assume that, at each time $i\ge1$, every agent $k$ observes a realization of some signal, $\bm \xi_{k,i}$, whose probability distribution is dependent on the true event $\theta^o$, namely, the process $\{\bm{\xi}_{k,i}\}$ is generated according to some known likelihood function $L_k(\cdot|\theta^{\circ} )$ -- see Fig.\ref{fig.observationGen}. We further assume that for each agent $k$, the signals $\{\bm\xi_{k,i}\}$ belong to a {\emph{finite}} signal space denoted by $Z_k$ and that these signals are independent over time.
	\begin{figure}[h!]
		\centering
		\includegraphics[scale=0.45]{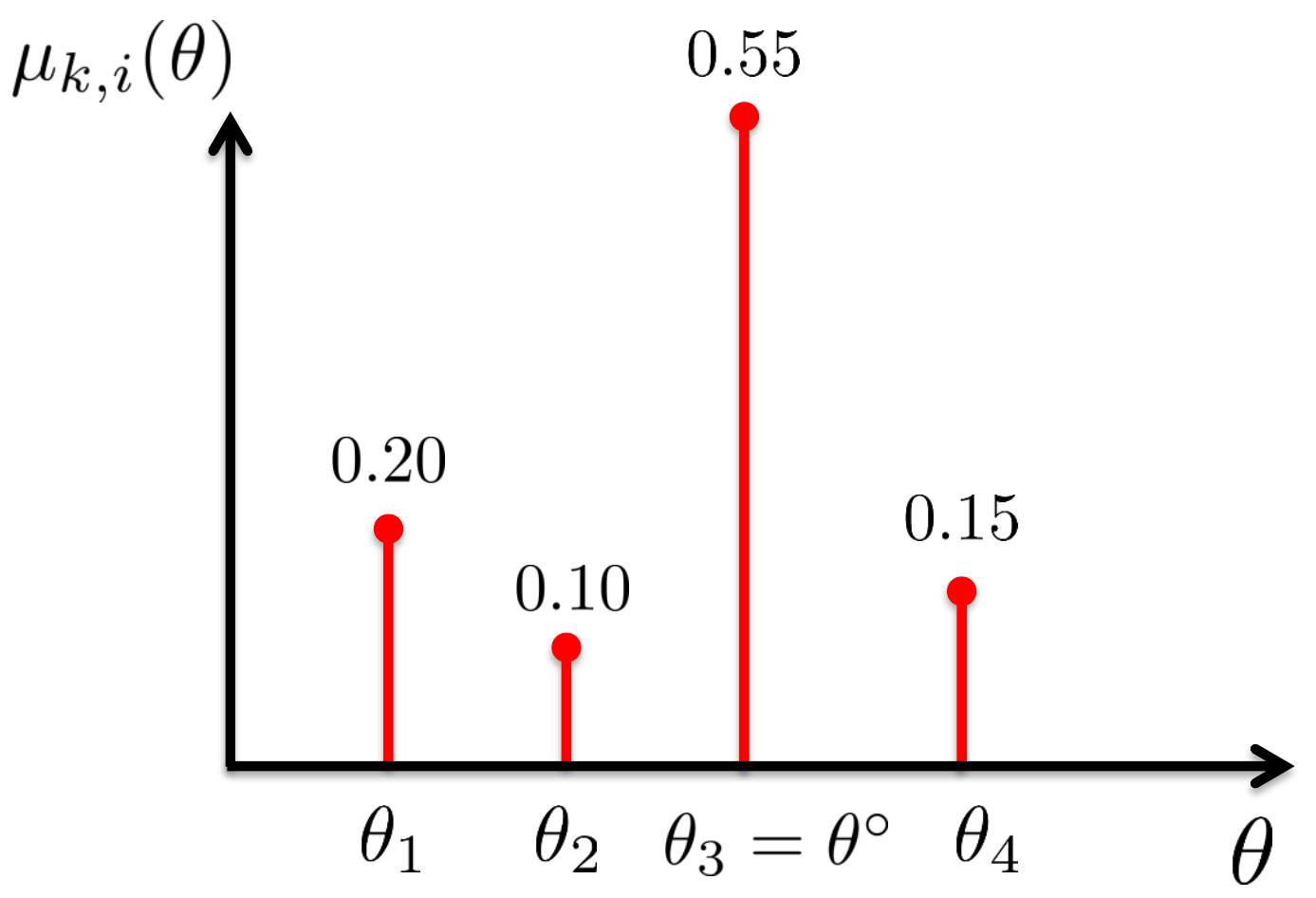} 
		\caption{\small An example of a belief distribution $\mu_{k,i}(\theta)$.}
		\label{fig.pmf}
	\end{figure}
	\begin{figure}[h!]
		\centering
		\includegraphics[scale=0.3]{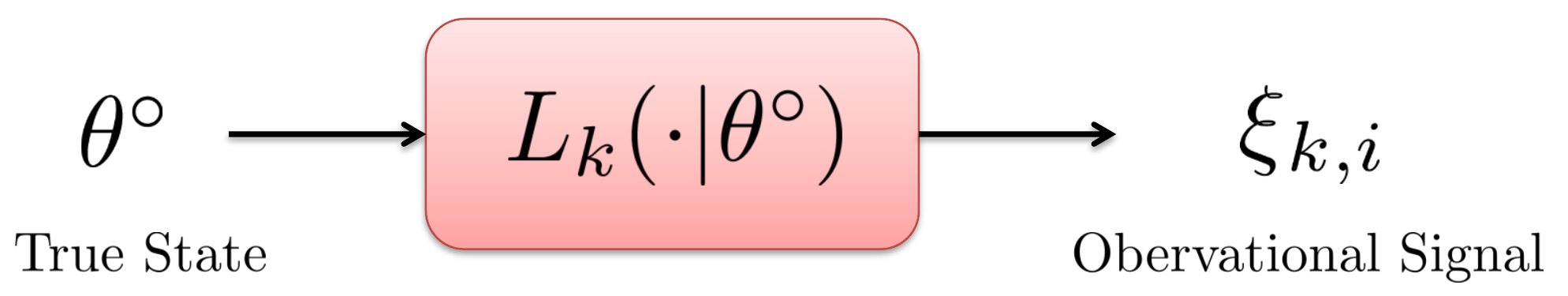} 
		\caption{\small Generation of observational signals.}
		\label{fig.observationGen}
	\end{figure}
	
	Diffusion social learning, described in \cite{zhao2012learning}, provides a mechanism by which agents can process the information they receive from their private signals and from their neighbors. A consensus-based strategy can also be employed, as was done in \cite{jadbabaie2012non}. We focus on the diffusion strategy due to its enhanced performance, as observed in \cite{zhao2012learning} and as further explained in the treatments \cite{Sayed,sayed2014adaptive}. In diffusion learning, at every time $i\ge1$, each agent $k$ first updates its belief, $
	\mu_{k,i-1}(\theta)$, based on its observed private signal $\xi_{k,i}$ by means of the Bayesian rule: 
	\be
	\label{eqn:bayesianupdate}
	\psi_{k,i}(\theta) = \frac{\mu_{k,i-1}(\theta)L_k(\xi_{k,i}|\theta)}{\sum_{\theta'\in\Theta}\mu_{k,i-1}
		(\theta')L_k(\xi_{k,i}|\theta')}
	\ee
	This step leads to an intermediate belief $\psi_{k,i}(\theta)$. After learning from their observed signals,
	agents can then learn from their social neighbors through cooperation to compute:
	\be
	\label{eqn:con_combine}
	\mu_{k,i}(\theta) = \sum_{\ell\in\mathcal{N}_k}a_{\ell k}\,\psi_{\ell,i}(\theta)
	\ee
	\noindent Subsequently, agent $k$ can use its updated belief, $\mu_{k,i}(\theta)$, to predict the probability of a certain signal $\zeta_k\in Z_k$ occurring in the next time instant $i+1$. This prediction or forecast is based on the following calculation:
	\be
	m_{k,i}(\zeta_k)\define\sum_{\theta\in\Theta}\mu_{k,i}
	(\theta)L_k(\zeta_{k}|\theta)=\mbox{\rm Prob}(\bm\xi_{k,i+1}=\zeta_k)
	\ee
	
	In the sequel, we will be interpreting the diffusion learning model as a {\em stochastic} system of interacting agents, especially since the operation of this mechanism is driven by the random observational signals. Thus, we rewrite  (\ref{eqn:bayesianupdate}) and (\ref{eqn:con_combine}) as follows by using boldface letters to refer to random variables.
	\be
	\left\{
	\begin{aligned}
		\bm\psi_{k,i}(\theta) & = \frac{\bmu_{k,i-1}(\theta)L_k(\bm\xi_{k,i}|\theta)}{\sum_{\theta'\in\Theta}\bmu_{k,i-1}
			(\theta')L_k(\bm\xi_{k,i}|\theta')} \\\\
		\bmu_{k,i}(\theta) & = \sum_{\ell\in\mathcal{N}_k}a_{\ell k}\,\bm\psi_{\ell,i}(\theta)
		\label{eqn:diffusion} 
	\end{aligned}
	\right.
	\ee

	\subsection{Correct Forecasting}
	When agents in strongly-connected networks follow model (\ref{eqn:diffusion}) to update their beliefs, the agents will eventually learn the truth according to the results established in \cite{zhao2012learning}. The argument there is based on an identifiability condition similar to the one used in \cite{jadbabaie2012non}, and which is motivated as follows. We assume first that the agents' private signals $\{\bm{\xi}_{k,i}\}$ do not hold enough information about the true state, so that individual agents cannot rely solely on their observations to identify $\theta^\circ$ and are motivated to cooperate. More specifically, this requirement amounts to assuming that each agent $k$ has a subset of states $\Theta_k \subseteq \Theta$ for which:
	\be 
	L_k(\zeta_k|\theta)=L_k(\zeta_k|\theta^\circ),\;\;\theta \in \Theta_k
	\label{eqn:setUndis}
	\ee
	for any $\zeta_k \in Z_k$. We refer to $\Theta_k$ as the set of indistinguishable states for agent $k$. We subsequently assume that through cooperation with their neighbors, agents are able to identify the true state by imposing the identifiability condition:
	\be
	\underset{k\in\mathcal{N}}{\bigcap}\Theta_k=\{\theta^\circ\} 
	\label{assm1}
	\ee
	We refer to this case as $\theta^\circ$ being globally identifiable. To prove that agents are able to learn the true state, the analysis in \cite{zhao2012learning} is based on first  showing that agents are able to learn the correct distribution of incoming signals.
	
	\begin{lemma}[Correct Forecasting \cite{zhao2012learning}] Assume that there exists at least one agent with a positive prior belief about the true state $\theta^\circ$, i.e., $\bmu_{k,0}(\theta^\circ)>0$ for some $k\in\mathcal{N}$. Then, agents are able to correctly predict the distribution of the incoming signals, namely, for any $\zeta_k \in Z_k$ and $k \in \mathcal{N}$:
		\be
		\lim_{i \to \infty} \bm{m}_{k,i}(\zeta_{k})\aseq L_k(\zeta_k|\theta^\circ)
		\label{CorrectForecast}
		\ee
	where $\aseq$ denotes almost-sure convergence. \qd
		\label{lemma1} 
	\end{lemma} 
	This lemma does not require the identifiability condition (\ref{assm1}). It explores forms of learning that were studied in \cite{jadbabaie2012non,epstein2010non} and also in \cite{Bayes,Bayes2}, which dealt with either learning the true parameter $\theta^\circ$ (similar to the setting we are considering) or learning the distribution of the incoming signal itself.
	
	Correct forecasting does not always imply the ability of agents to learn the true parameter, $\theta^\circ$. However, in the case of strongly-connected networks, this conclusion is true under some conditions mentioned next (the same implication will not hold for weakly-connected networks; there, we will show that correct forecasting does not imply the ability of agents to learn the truth). 
	
	\begin{theorem}[Truth Learning \cite{zhao2012learning}]\label{theo.1}
		Under the same conditions of Lemma \ref{lemma1}, assume that there exists at least one prevailing signal $\zeta_k^\circ$ for each agent $k$, namely, that
		\be
		L_k(\zeta_k^\circ|\theta^\circ)-L_k(\zeta_k^\circ|\theta) \geq 0, \quad \forall \theta \in \Theta\setminus\Theta_k
		\label{prev}
		\ee
		and assume as well that the true state $\theta^\circ$ is globally identifiable as in (\ref{assm1}). Then, all agents asymptotically learn the truth, i.e., for any $k\in\mathcal{N}$:
		\be
		\lim_{i\rightarrow\infty} \bm \mu_{k,i}(\theta^\circ) \aseq 1
		\label{truthLearning}
		\ee
	 \qd
		\label{theorem1}
	\end{theorem}
	Figure \ref{fig.prevailing} illustrates what it means for a prevailing signal to exist for an agent $k$. In this example, the true state $\theta^\circ$ is assumed to be $\theta_1$. Assume also that for agent $k$, the set of distinguishable states is $\bar\Theta_k= \Theta\setminus\Theta_k=\{\theta_2,\theta_3\}$ and the space of observational signals is $Z_k=\{\zeta_1,\zeta_2,\zeta_3,\zeta_4\}$. We see in the example that the signal $\zeta_1$ plays the role of a prevailing signal. This is because when the true state is $\theta_1$, the likelihood of $\zeta_1$ is greater than its likelihood when the true state is $\theta_2$ or $\theta_3$, i.e.,
	\be
	L_k(\zeta_1|\theta_1)>L_k(\zeta_1|\theta_2) \text{, }  L_k(\zeta_1|\theta_1)>L_k(\zeta_1|\theta_3)
	\ee
	These two conditions are not jointly satisfied for the other observational signals. The presence of a prevailing signal provides agent $k$ with sufficient information to identify the distinguishable set $\bar\Theta_k=\Theta\setminus \Theta_k$. This means that agent $k$ will be able to assign a zero probability to any $\theta$ in this set. Then, with the help of neighboring agents, and in the presence of the identifiability condition (\ref{assm1}), agent $k$ will be able to discover the true sate $\theta^\circ$ in $\Theta_k$. 
	\begin{figure}[h!]
		\centering
		\includegraphics[scale=0.45]{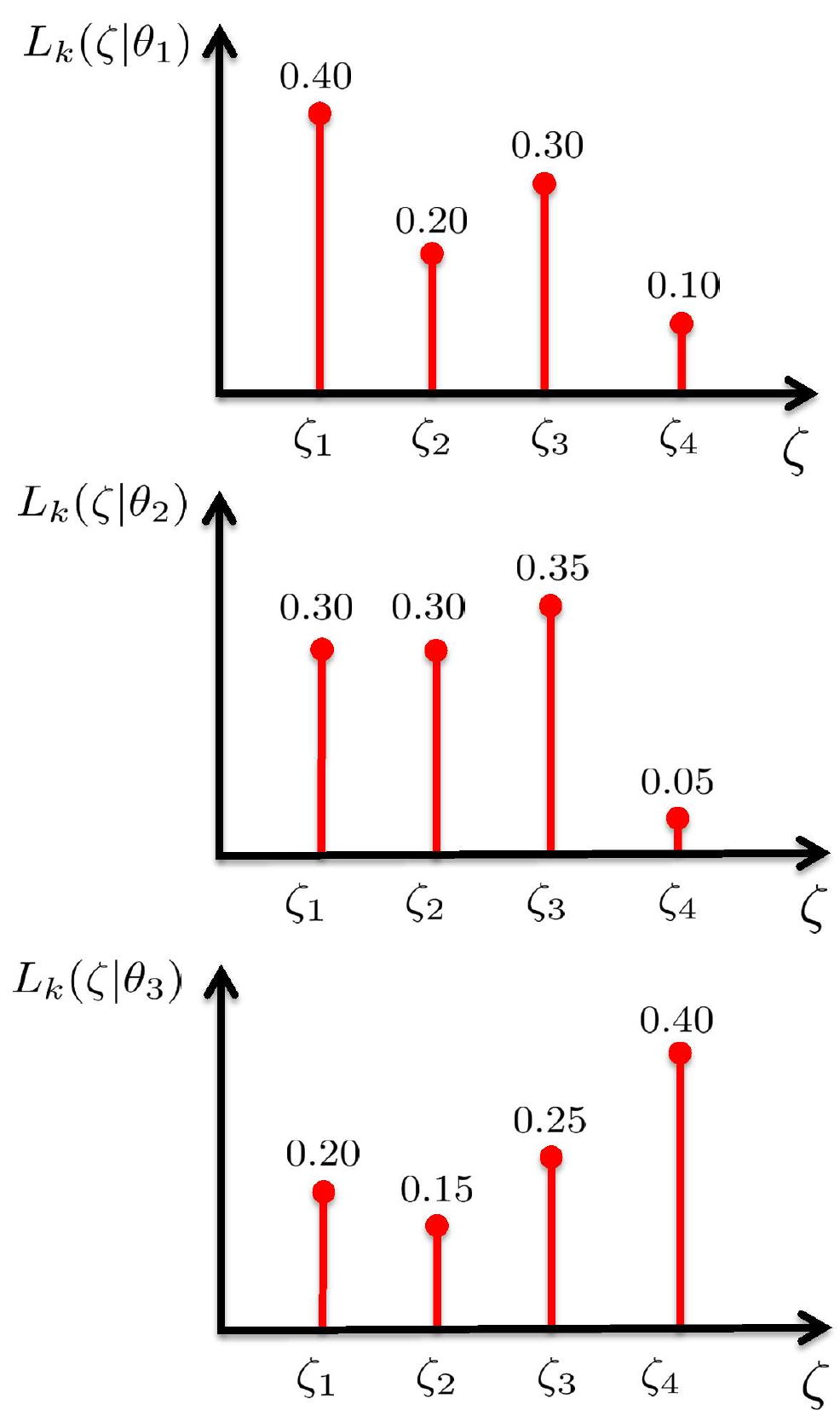} 
		\caption{\small An example showing the existence of a prevailing signal $\zeta_1$ for agent $k$.}
		\label{fig.prevailing}
	\end{figure}
	\section{Diffusion Learning over Weak Graphs}
	We now examine how learning is affected when the agents are connected by a weak topology. In simple terms, weakly-connected networks consist of multiple sub-networks where at least one sub-network feeds information forward to some other sub-network but never receives information back from this sub-network. The example in Fig. \ref{fig.WCFigure} illustrates a situation involving four sub-networks. The agents in each sub-network observe signals related to their own true states denoted by $\theta^\circ_1$, $\theta^\circ_2$, $\theta^\circ_3$, and  $\theta^\circ_4$. For generality, we do not require the true states to be the same across the sub-networks. In the figure, each of the two sub-networks on top is strongly-connected. Therefore, if their agents follow the diffusion social learning model (\ref{eqn:diffusion}), they can asymptotically learn their true states. The third and fourth sub-networks in the bottom receive information from the top sub-networks without feeding information back to them. As the analysis will show, this structure results in the two top sub-networks playing the role of influential entities that impose their beliefs on the agents in the bottom sub-networks, regardless of the local information that is sensed by these latter sub-networks.
	\subsection{Weak Network Model}
	We first review the main features of the weakly-connected network model from \cite{ying2014information, icassp}.
	Consider a network that consists of two types of sub-networks: $S$ sub-networks and $R$ sub-networks. Each sub-network in the $S$ family has a strongly-connected topology. In contrast, each sub-network in the $R$ family is only required to be connected. This means that any receiving sub-network has a path connecting any two agents without requiring any agent to have a self-loop.  Moreover, the interaction between $S$ and $R$ sub-networks is not symmetric: information can flow from $S$ (``sending'') sub-networks to $R$ (``receiving'') sub-networks but not the other way around. We index each strongly-connected sub-network by $s$ where $s=\{1,2,\cdots,S\}$. Similarly, we index each receiving sub-network by $r$ where $r=\{S+1,\cdots,S+R\}$. Each sub-network $s$ has $N_s$ agents, and the total number of agents in the $S$ sub-networks is:
	\be
	N_{gS}\define N_1 +N_2+\dots+N_S
	\ee
	Similarly, each sub-network $r$ has $N_{r}$ agents, and the total number of agents in the $R$ sub-networks is:
	\be
	N_{gR}\define N_{S+1} +N_{S+2}+\dots+N_{S+R}
	\ee
	We still denote by $N$ the total number of agents across all sub-networks, i.e., $N=N_{gS}+N_{gR}$. 
	We continue to denote by $\mathcal{N}=\{1,2,\cdots,N\}$ the indexes of the agents. We assume that the agents are numbered such that the indexes of $\mathcal{N}$ represent first the agents from the $S$ sub-networks, followed by those from the $R$ sub-networks. In this way, the structure of the network is represented by a large $N \times N$ combination matrix $A$, which will have an upper block-triangular structure of the following form \cite{ying2014information, icassp}: 
	
	\begin{footnotesize}
		\begin{eqnarray}
			\nonumber
			\begin{array}{cc}
				\overbrace{\rule{30mm}{0mm}}^{\mathrm{Subnetworks:} 1,2,\ldots, S} &\; \overbrace{\rule{50mm}{0mm}}^{\mathrm{Subnetworks:} S+1, S+2, \ldots,S+R}
			\end{array}
			\\ \nonumber
			\left[
			\begin{array}{cccc|cccc}
				A_{1} 	&	0 	     	&\hdots	& 0 		&	A_{1,S+1} 		& A_{1,S+2}		 &\hdots		 &A_{1,S+R}	\\
				0          	& A_{2}   		&\hdots	& 0 		& 	A_{2,S+1} 		& A_{2,S+2}	 	 &\hdots		 &A_{2,S+R}	\\
				\vdots 	& \vdots		&\ddots	&\vdots	& 	\vdots 			& \vdots			 &\ddots		 &\vdots	 \\
				0          	& 	0  		&\hdots	& A_{S} 	& 	A_{S,S+1} 		& A_{S,S+2}	 	 &\hdots		 &A_{S,S+R}	\\
				\hline
				0          	& 	0  		&\hdots	& 0 		& 	A_{S+1} 			& A_{S+1,S+2}	 	 &\hdots		 &A_{S+1,S+R}	\\
				0          	& 	0  		&\hdots	& 0 		& 	0		 		& A_{S+2}	 	 	 &\hdots		 &A_{S+2,S+R}	\\
				\vdots 	& \vdots		&\ddots	&\vdots	& 	\vdots 			& \vdots			 &\ddots		 &\vdots	 \\
				0          	& 	0  		&\hdots	& 0 		& 	0				& 0			 	 &\hdots		 &A_{S+R}	 \\
			\end{array}
			\right]
		\end{eqnarray}
	\end{footnotesize}\vspace{-0.32cm}
	\be
	\label{matrixA}
	\ee 
	\begin{figure}[h!]
		\centering
		\includegraphics[scale=0.35]{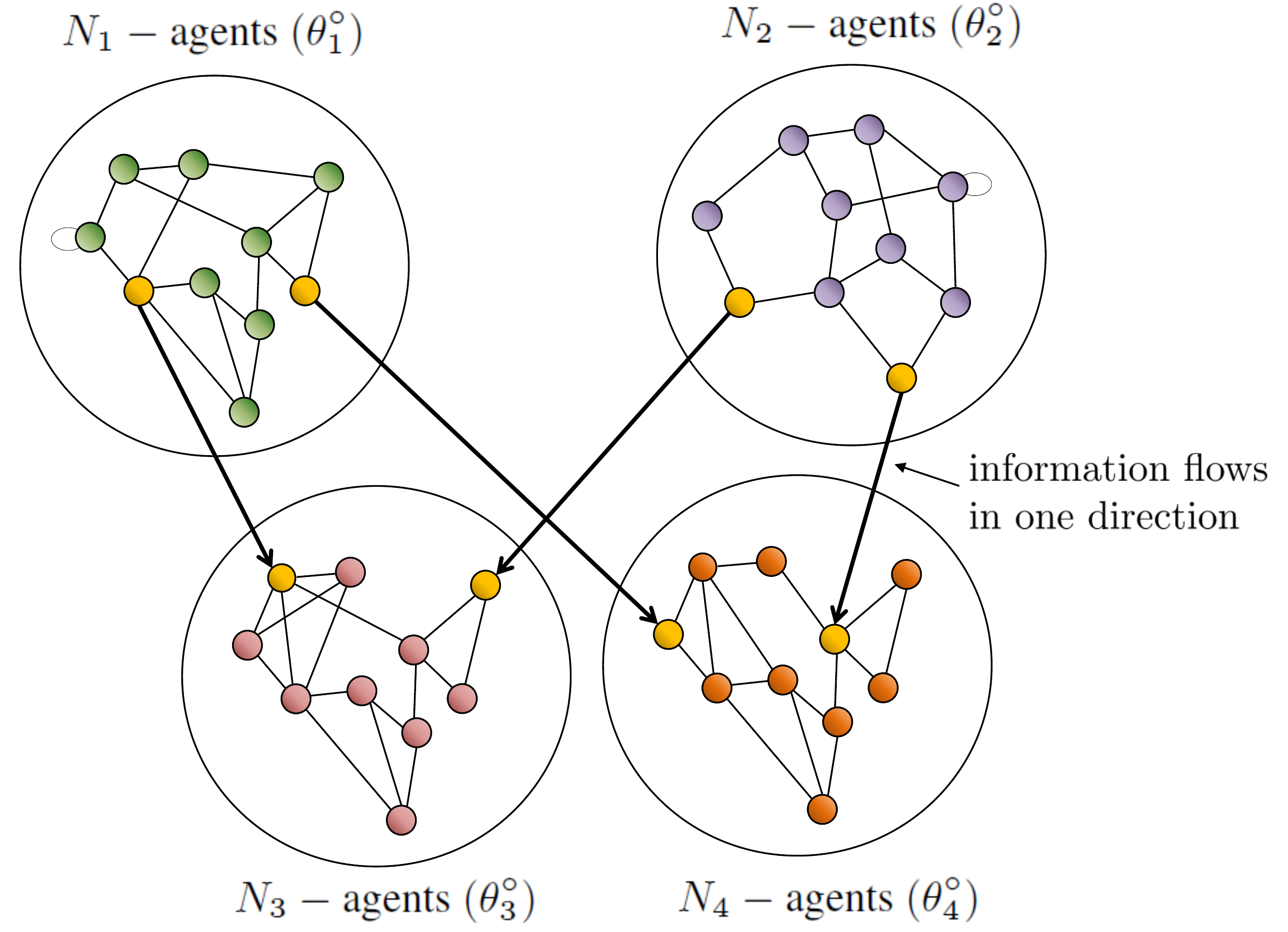} 
		\caption{\small An example of a weakly connected network.}
		\label{fig.WCFigure}
	\end{figure}
	
	The matrices $\{A_1,\cdots,A_S\}$ on the upper left corner are left-stochastic primitive matrices corresponding to the $S$ strongly-connected sub-networks. Each of these matrices has spectral radius equal to one, $\rho(A_s)=1$. Moreover, each $A_s$ has a single eigenvalue at one and the corresponding right eigenvector has positive entries. We denote it by $y_s$ and normalize its entries to add up to one, i.e., $\one\tran y_s=1$. 
	
	Likewise, the matrices $\{A_{S+1},\cdots,A_{S+R}\}$ in the lower right-most block correspond to the internal weights of the $R$ sub-networks. These matrices are not necessarily left-stochastic because they do not include the coefficients over the links that connect the $R$ sub-networks to the $S$ sub-networks. Nevertheless, based on results from \cite{meyer2000matrix}, it was shown in \cite{ying2014information} that for any receiving subnetwork $r$, it holds that $\rho(A_r)<1$. Moreover, since $A_{r}$ has non-negative entries and sub-network $r$ is connected, it follows from the Perron-Frobenius theorem \cite{perron, matrixAnalysis} that $A_r$ has a unique positive real eigenvalue $\lambda_r$, that is equal to its spectral radius $\rho(A_r)$, and the corresponding right eigenvector has positive entries. We denote this eigenvector by $y_r$. We again normalize the entries of $y_r$ to add up to one, $\one\tran y_r=1$:
	\be
	A_{r}y_r=\lambda_ry_r,\quad \one\tran y_r=\one,\quad y_r \succ 0
	\label{perronFrob2}
	\ee
	We denote the block structure of $A$ in (\ref{matrixA}) by:
	\be
	A \define \ba{ccc}T_{SS}&\vline&T_{SR}\\\hline 0&\vline &T_{RR}\ea \label{AStruct}
	\ee 
	This specific structure has one useful property that we will exploit in the analysis.
	
	\begin{lemma}[Limiting Power of $A$ \cite{ying2014information}] It holds that:
		\begin{equation}
			A_{\infty} \stackrel {\triangle}{=}
			\lim_{n\to\infty} A^n =
			\left[
			\begin{array}{ccc}
				E &\vline& E W \\\hline
				0 &\vline& 0 \\
			\end{array}
			\right]
			\label{label.eq30}\end{equation}
		where the $N_{gS} \times N_{gS} $ matrix $E$ and the $N_{gS} \times N_{gR}$ matrix $W$ are given by:
		\bq
		W &\define & T_{SR}(I-T_{RR})^{-1} \label{eq:defW}\\
		E &\define &\mbox{\rm blockdiag}\left\{ y_1\one_{N_1}\tran,\ldots,y_S\one_{N_S}\tran\right\}
		\eq
		The matrix $W$ has non-negative entries and the sum of the entries in each column is equal to one.
		\label{kajd713.lemma}
		\qd
	\end{lemma}
	
	We now examine the belief evolution of agents in weakly-connected networks. We still denote by $\Theta$ the set of all possible states, and we assume that $\Theta$ is uniform across all sub-networks. However, we allow each sub-network to have its own true state, which may differ from one sub-network to another. We denote by $\theta^\circ_s$ the true state of sending sub-network $s$ and by $\theta_r^\circ$ the true state of receiving sub-network $r$, where both $\theta^\circ_s$ and $\theta^\circ_r$ are in $\Theta$. Therefore, if agent $k$ belongs to a sub-network $s$, its observational signals $\bm\xi_{k,i}$ will be generated according to the likelihood function $L_k(\cdot|\theta^{\circ}_s )
	$. On the other hand, if agent $k$ belongs to a sub-network $r$, its observational signals $\bm\xi_{k,i}$ will be generated according to $L_k(\cdot|\theta^{\circ}_r )
	$.
	
	We already know that the $S-$type sub-networks are strongly-connected, so that their agents can cooperate together to learn the truth. More specifically, according to Theorem \ref{theorem1}, if agent $k$ belongs to sub-network $s$, then it holds that:
	\be
	\lim_{i \to \infty} \bm\mu_{k,i}(\theta^\circ_s)\aseq1 
	\label{eqn:delta}
	\ee
	The question that we want to examine is how the beliefs of the agents in the receiving sub-networks are affected. These agents are now influenced by the beliefs of the $S-$type groups. Since this external influence carries information not related to the true state of each receiving sub-network, the receiving agents may not be able to learn their own true states. We will show that a leader-follower relationship develops.
	\subsection{Diffusion Social Learning over Weak Graphs}
	We consider that all agents are following the diffusion strategy (\ref{eqn:diffusion}) for social learning. In a manner similar to (\ref{eqn:setUndis}), if agent $k$ belongs to sub-network $r$, then we assume that there exists a subset of states $\Theta_k \subseteq \Theta$ such that:
	\be 
	L_k(\zeta_k|\theta)=L_k(\zeta_k|\theta_r^\circ)
	\label{indisSet}
	\ee
	for any $\zeta_k \in Z_k$ and $\theta \in \Theta_k$, i.e., $\Theta_k$ is the set of indistinguishable states for agent $k$. Moreover, we assume a scenario in which the private signals of agents in the receiving sub-networks are not informative enough to let their agents discover that the true states of the sending sub-networks do not represent their own truth. That is, we are assuming \emph{for now} the following condition.
	\begin{assumption*}
		The true state $\theta^\circ_s$, of each sub-network $s\in\{1,2,\ldots,S\}$, belongs to the indistinguishable set $\Theta_k$:
		\be
		\theta^\circ_s \in \Theta_k, \quad \mbox{\rm for any $k>N_{gS}$}
		\label{assumImp}
		\ee 
		\qd
	\end{assumption*}
	
	Under (\ref{assumImp}), we will now verify that the interaction with the $S$ sub-networks ends up forcing the receiving agents to focus their beliefs on the true states of the $S-$type. Later, we will show that a similar conclusion continues to hold even when (\ref{assumImp}) is relaxed. 
	
	Thus, let $\Theta^\bullet=\{\theta^\circ_1,\cdots,\theta_S^\circ\}$ denote the set of all true states of the $S-$type sub-networks. We are assuming, for notational simplicity, that the true states $\{\theta^\circ_s\}$ are distinct from each other. Otherwise, we only include in $\Theta^{\bullet}$ the set of truly distinct states, which will be smaller than $S$ in number. We denote the complement of $\Theta^\bullet$ by $\bar{\Theta}^\bullet$, such that $\Theta^\bullet\cap\bar{\Theta}^\bullet=\emptyset$ and $\Theta^\bullet\cup\bar{\Theta}^\bullet=\Theta$. We first show that as $i\to\infty$, each receiving agent $k$ will assign zero belief to any event $\theta\in\bar{\Theta}^\bullet$. This means that receiving agents will end up searching for the truth within the set $\Theta^{\bullet}$. 
	\begin{lemma}[Focus on True States of $S$ Sub-Networks]
		\label{lemma:vanishing}
		Under (\ref{assumImp}), each agent $k$ of any receiving sub-network $r$ eventually identifies the set $\bar{\Theta}^\bullet$, namely, for any $\theta \in \bar{\Theta}^\bullet$:
		\be
		\lim_{i \to \infty} \bm\mu_{k,i}(\theta)\aseq 0 
		\ee
	\end{lemma}
	\noindent {\em Proof}: See Appendix \ref{App.B}.
	\qd
	
	This lemma implies that the receiving agents are still able to perform correct forecasting.  
	\begin{lemma}[Correct Forecasting]
		\label{lemma:correct}
		Under (\ref{assumImp}), every agent $k$ in sub-network $r$ develops correct forecasting, namely,
		\be
		\lim_{i \to \infty} \bm{m}_{k,i}(\zeta_{k})\aseq L_k(\zeta_k|\theta_r^\circ),\quad \mbox{\rm for any $\zeta_k \in Z_k$}
		\ee
	\end{lemma}
	\noindent {\em Proof}: See Appendix \ref{App.B1}.
	\qd
	
	Even with the external influence, agent $k$ is still able to attain correct forecasting because any true state $\theta^\circ_s$ of any sending sub-network, $s$, belongs to the indistinguishable set of agent $k$, i.e.,  $L_k(\zeta_k|\theta^\circ_r)=L_k(\zeta_k|\theta^\circ_s)$ from (\ref{assumImp}) and (\ref{indisSet}). Since agents zoom onto the set $\Theta^\bullet$, this fact enables correct forecasting but does not necessarily imply truth learning for weak graphs, as discussed in the sequel.
	
	The previous two lemmas establish that the belief of each agent $k$ in sub-network $r$ will converge to a distribution whose support is limited to $\theta\in\Theta^{\bullet}$. The next question is to evaluate this distribution, which is the subject of the following main result. First let
	  \begin{align}
	   	\hspace{-0.33cm}\bmu_{i}^s (\theta)\define&\;
	   	{\rm col}\left\{
	   	\bmu_{k_s(1),i} (\theta),
	   	\bmu_{k_s(2),i} (\theta),
	   	\hdots,
	   	\bmu_{k_s{(N_s)},i}(\theta) 
	   	\right\} \\
	   	\hspace{-0.33cm}\bmu_{i}^r (\theta)\define&\;
	   	{\rm col} \left\{
	  	\bmu_{k_r(1),i} (\theta),
	  	\bmu_{k_r(2),i} (\theta),
	   	\hdots,
	  	\bmu_{k_r(N_{r}),i}(\theta)
	   	\right\}
	   	\label{muR}
	   \end{align}
	collect all beliefs from agents that belong respectively to sub-network $s$ and sub-network $r$, where the notation $k_s(n)$ denotes the index of the $n${-th} agent within sub-network $s$, i.e.,
	\be 
	k_s(n)=\sum\limits_{v=1}^{s-1}N_v+n
	\ee
	and $n\in \{1,2,\cdots,N_s\}$ and the notation $k_r(n)$ denotes the index of the $n$-th agent within sub-network $r$, i.e.,
	\be
	k_r(n)=N_{gS}+\sum\limits_{v=S+1}^{r-1}N_{v}+n
	\label{Index}
	\ee
	and $n\in \{1,2,\cdots,N_r\}$. Furthermore, let
	\begin{align}
	  	\bmu_{\mathcal{S},i} (\theta)\define&\;
	  	{\rm col}\left\{
	  	\bmu_{i}^1 (\theta),
	  	\bmu_{i}^2 (\theta),
	  	\hdots,
	  	\bmu_{i}^S (\theta)
	  	\right\} \\
	  	\bmu_{\mathcal{R},i} (\theta)\define&\;
	  	{\rm col} \left\{
	  	\bmu_{i}^{S+1} (\theta) ,
	  	\bmu_{i}^{S+2} (\theta) ,
	  	\hdots,
	  		\bmu_{i}^{S+R}(\theta) 
	  	\right\}
	  	\label{muS}
	\end{align}
	collect all belief vectors respectively from all $S-$type sub-networks and from all $R-$type sub-networks.
	\begin{theorem}[Limiting Beliefs for Receiving Agents]
		\label{mainresult1}
		Under (\ref{assumImp}), it holds that
		\be
		\lim_{i \to \infty} \bmu_{\mathcal{R},i} (\theta) = W\tran \left (\lim_{i \to \infty}  \bmu_{\mathcal{S},i} (\theta)\right) 
		\label{result}
		\ee
	\end{theorem}
	\noindent {\em Proof}: See Appendix \ref{App.B2}.
	\qd
	
	We expand (\ref{result}) to clarify its meaning and to show how the beliefs are distributed among the elements of $\Theta^\bullet$.
	We already know from the result in Theorem \ref{theorem1} that, for each agent $k$ of sending sub-network $s$, $\bm{\mu}_{k,i}(\theta)$ converges asymptotically to an impulse of size one at the location $\theta=\theta^\circ_s$. Thus, we write:
	\be
	\lim_{i\to\infty}\bm{\mu}_{i}^s(\theta)=
	e_{\theta,\theta^\circ_s}
	\define \left\{
	\begin{aligned}
		\one_{N_s},&\quad {\rm if}\quad \theta = \theta_s^\circ\\
		\boldsymbol{0}_{N_s},&\quad {\rm otherwise}
	\end{aligned}
	\right. 
	\label{eqn:delta2}
	\ee
	where $\one_{N_s}$ denotes a column vector of length $N_s$ whose elements are all one. Similarly, $\boldsymbol{0}_{Ns}$ denotes a column vector of length $N_s$ whose elements are all zero. Hence,
	\be
	\lim_{i \to \infty} \bm\mu_{\mathcal{S},i}(\theta)=
	{\rm col}\left\{
	e_{\theta,\theta^\circ_1},
	e_{\theta,\theta^\circ_2},
	\hdots,
	e_{\theta,\theta^\circ_S}
	\right\}
	\ee
	Now, let $w_{k}\tran$ denote the row of $W\tran$ that corresponds to agent $k$ in sub-network{\footnote{The real index of the row of $W\tran$ that corresponds to agent $k$ is $k-N_{gS}$.} $r$. We partition it into
		\be
		w_k\tran=\ba{cccc}w_{k,N_1}\tran&w_{k,N_2}\tran&\ldots,w_{k,N_S}\tran\ea
		\ee
		where the $\{N_1,N_2,\ldots,N_S\}$ are the number of agents in each sub-network $s\in\{1,2,\ldots,S\}$.
		By examining (\ref{result}), we conclude that the distribution for each agent $k$ in an $R-$type
		sub-network converges to a combination of the various vectors $\{e_{\theta,\theta^\circ_s}\}$, namely,
		\vspace{-0.5cm}
		\begin{align}
		\lim_{i \to \infty} \bm\mu_{k,i}(\theta)=q_k(\theta)\define \sum_{s=1}^{S}w_{k,N_s}\tran e_{\theta,\theta^\circ_s}
		\label{FinalDistribution}
		\end{align}
		\noindent Observe that, from this equation, to get $q_k(\theta^\circ_s)$, the elements of the corresponding block in $w_k$, i.e., $w_{k,N_s}$, should be summed. Now, if we consider that multiple sending sub-networks have the same true state, then to get $q_k(.)$ at this true state, the elements of all corresponding blocks in $w_k$ will need to be summed. Note that this is a valid probability measure in view of Lemma \ref{kajd713.lemma}, i.e.,
		\bq
		\sum_{\theta \in \Theta^\bullet}q_k(\theta)= 1
		\eq
		Note also that if it happens that $\theta_s^\circ=\theta^\circ$ for all $s$, then $q_k(\theta^\circ)=1$ and
		$q_k(\theta)=0$ for all $\theta\neq \theta^\circ$, and in this case, sending agents can be seen as helping receiving agents to find the true state. We also observe that the beliefs of agents in the receiving sub-networks differ from one agent to another, since for each agent $k$, $q_k(\theta)$ depends on $w_k$. This means that the external influence has created social disagreement in the receiving sub-networks.
		
		We therefore established that the beliefs of receiving agents converge to a distribution whose support is limited to the true states of the sending sub-networks. We will refer to this situation as a {\em total influence} or ``mind-control'' scenario where the learning of the $R-$subnetworks is fully dictated by the $S-$subnetworks. When all agents follow model ({\ref{eqn:diffusion}}) and when assumption (\ref{assumImp}) is satisfied, this total influence scenario arises. Although the private signals of the receiving agents are supposed to hold information regarding their own true state, however, under assumption (\ref{assumImp}), these signals are not informative enough, so that agents are naturally driven to be under the influence of  the sending sub-networks. 
		
		We are interested now in knowing whether this total influence situation can still occur when assumption (\ref{assumImp}) is not satisfied anymore. When this is the case, sending agents may not be able to totally control the beliefs of receiving agents anymore. Before establishing the analytical results, and before showing how self-awareness can alter this dynamics, we provide an illustrative example. 
		
		\subsection{Implications of Violating Condition  (\ref{assumImp})}
		We consider a network consisting of three agents, with the first two playing the role of influential agents and the third one acting as a receiving agent. The combination matrix is chosen as follows:
		
		\begin{small}
			\begin{equation}
				A=\left[
				\begin{array}{cc|c}
					1 & 0 & 0.1 \\
					0 & 1 & 0.2\\
					\hline
					0 & 0 & 0.7
				\end{array}
				\right]
				\label{examComb}
			\end{equation}
		\end{small}
		\begin{figure}[h!]
			\centering
			\includegraphics[scale=0.4]{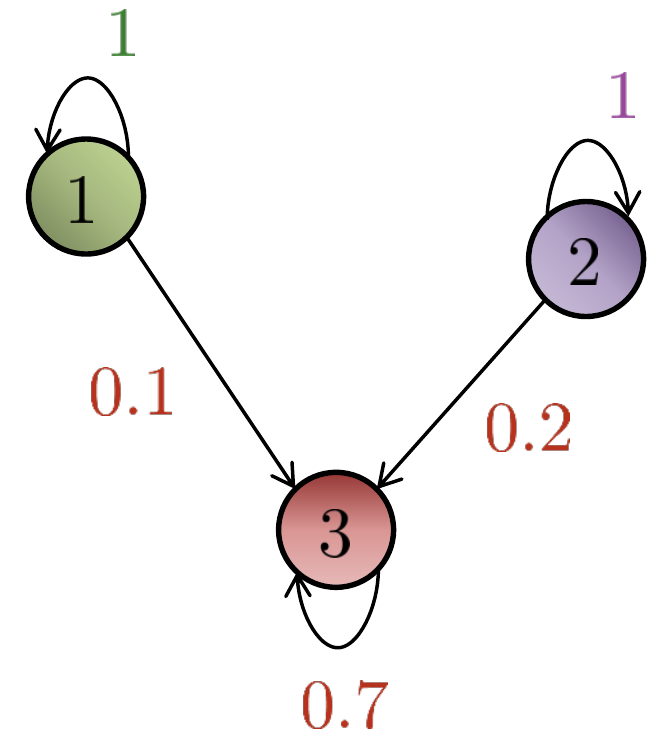} 
			\caption{{\small A weakly connected network and the corresponding combination policy (\ref{examComb}). }}
			\label{network.counter}
		\end{figure}
		
		\noindent We denote by $\theta^\circ_1$ the true state for agent $1$, by $\theta^\circ_2$ the true state for agent $2$, and by $\theta^\circ_3$ the true state for agent $3$ so that $\Theta=\{\theta^\circ_1,\theta^\circ_2,\theta^\circ_3\}$. The observational signal for all three agents is either a head ``H'' or a tail ``T''. In order for agents $1$ and $2$ to learn their true states asymptotically, we need to ensure that the
		conditions of Theorem \ref{theorem1} are satisfied. One of these requirements is the identifiability condition (\ref{assm1}), which requires that the intersection of the indistinguishable sets (\ref{eqn:setUndis}) of all agents in a given sending sub-network $s$ must be the singleton $\{\theta^\circ_s\}$. In this example, each sending sub-network consists of only one agent, so that condition (\ref{assm1}) reduces to $\Theta_1=\{\theta^\circ_1\}$ for the first agent and $\Theta_2=\{\theta^\circ_2\}$ for the second agent. In other words, since agents $1$ and $2$ do not have neighbors to communicate with, they must rely solely on their observational signals to learn the truth. This is feasible when for agents $1$ and $2$ no state is observationally equivalent to their true state (or indistinguishable). Using the definition of the indistinguishable set (\ref{eqn:setUndis}), $\Theta_1=\{\theta^\circ_1\}$ translates into the following requirement for agent $1$:
		\be 
		L_1(\zeta_1|\theta^\circ_1)\neq L_1(\zeta_1|\theta^\circ_2) \; \text{ and } \; L_1(\zeta_1|\theta^\circ_1)\neq L_1(\zeta_1|\theta^\circ_3)
		\label{condi1}
		\ee
		for any $\zeta_1 \in\{H,T\}$. Similarly, $\Theta_2=\{\theta^\circ_2\}$ translates into the following requirement for agent $2$:
		\be 
		L_1(\zeta_2|\theta^\circ_2)\neq L_1(\zeta_2|\theta^\circ_1) \; \text{ and } \; L_1(\zeta_2|\theta^\circ_2)\neq L_1(\zeta_2|\theta^\circ_3)
		\label{condi2}
		\ee
		for any $\zeta_2 \in\{H,T\}$. For this example, we are choosing the likelihood functions arbitrarily but satisfying (\ref{condi1}) for agent $1$ and (\ref{condi2}) for agent $2$. For instance, we select for agent $1$,
		\bq
		L_1(H|\theta^\circ_1)=0.10, \; L_1(H|\theta^\circ_2)=0.35, \; L_1(H|\theta^\circ_3)=0.45 
		\eq
		and set $L_1(T|\theta)=1-L_1(H|\theta)$ for any $\theta \in\Theta$. Likewise for agent $2$, we select
		\bq
		L_2(H|\theta^\circ_1)=0.10, \; L_2(H|\theta^\circ_2)=0.20, \; L_2(H|\theta^\circ_3)=0.30 
		\eq
		and set $L_2(T|\theta)=1-L_2(H|\theta)$ for any $\theta \in\Theta$. Before analyzing the beliefs of agent $3$ when (\ref{assumImp}) is not satisfied, we consider first the case in which this assumption is satisfied. In this way, we will be able to compare what is happening in both cases. More specifically, following (\ref{assumImp}), we consider first that $\theta^o_1$ and $\theta^o_2$ belong to the indistinguishable set of agent $3$ denoted by $\Theta_3$, i.e., $\{\theta^\circ_1,\theta^\circ_2\}\in\Theta_3$. This means, according to the definition of the indistinguishable set (\ref{eqn:setUndis}), that
		\be
		L_3(\zeta_3|\theta^\circ_1)=L_3(\zeta_3|\theta^\circ_3) \;\text{ and }\; L_3(\zeta_3|\theta^\circ_2)=L_3(\zeta_3|\theta^\circ_3)
		\ee 
		for any $\zeta_3\in\{H,T\}$.
		According to model (\ref{eqn:diffusion}), the intermediate belief of agent $3$ is given by:
		\begin{align}
		\bm\psi_{3,i}(\theta) 
		&=\frac{\bmu_{3,i-1}(\theta)L_3(\bm\xi_{3,i}|\theta)}{\left(\sum_{\theta'\in\Theta}\bmu_{3,i-1}
			(\theta')\right)L_3(\bm\xi_{3,i}|\theta)}  \nn \\
		&=\bmu_{3,i-1}(\theta)
		\end{align}
		We observe in this example that the private signals of agent $3$ end up not contributing to its intermediate belief. As a result, it is only the beliefs of agents $1$ and $2$ that affect the belief of agent $3$, so that:
		\begin{align}
		\bmu_{3,i}(\theta)&= a_{13}\bpsi_{1,i}(\theta)+a_{23}\bpsi_{2,i}(\theta)+a_{33}\bpsi_{3,i}(\theta) \nn \\
		{}&=a_{13}\bmu_{1,i}(\theta)+a_{23}\bmu_{2,i}(\theta)+a_{33}\bmu_{3,i-1}(\theta) 
		\label{noContri}
		\end{align}
		In writing (\ref{noContri}), we used the fact that the intermediate beliefs for agents $1$ and $2$ coincide with their updated beliefs since, in this example, agents $1$ and $2$ have no neighbors. Thus, since $a_{33}<1$,
		\be
		\lim_{i\to\infty}\bmu_{3,i}(\theta)= \nn 
		\ee
		\be
		\left(\frac{a_{13}}{1-a_{33}}\right)\lim_{i\to \infty}\bmu_{1,i}(\theta)+\left(\frac{a_{23}}{1-a_{33}}\right)\lim_{i\to\infty}\bmu_{2,i}(\theta)
		\ee
		from which we conclude that
		\begin{align}
		\lim_{i\to\infty}\bmu_{3,i}(\theta^\circ_1)&=\frac{a_{13}}{1-a_{33}}  \\ \lim_{i\to\infty}\bmu_{3,i}(\theta^\circ_2)&=\frac{a_{12}}{1-a_{33}}  \\
		\lim_{i\to\infty} \bmu_{3,i}(\theta^\circ_3)&= 0
		\end{align}
		This total influence result is expected to occur according to Theorem \ref{mainresult1}, when assumption (\ref{assumImp}) is satisfied. 
		
		Let us consider now the case in which assumption (\ref{assumImp}) is not satisfied. This means that $\theta^\circ_1$ and $\theta^\circ_2$ do not need to both belong to the indistinguishable set $\Theta^\circ_3$ of agent 3, i.e.,
		\begin{align} 
		L_3(\zeta_k|\theta^\circ_1) \neq  L_3(\zeta_k|\theta^\circ_3) \quad\text{or}\quad 	L_3(\zeta_k|\theta^\circ_2) \neq  L_3(\zeta_k|\theta^\circ_3)
		\label{condi11}
		\end{align}
		for any $\zeta_k \in \{H,T\}$. In this example, we study the worst case scenario in which both conditions in (\ref{condi11}) are met (even if we consider other situations in which only one of these conditions is met, we still arrive at a similar conclusion, namely, the belief of agent $3$ will not reach a fixed distribution). We select arbitrarily the values for the likelihood function of agent $3$, but in a way that these values satisfy both conditions in (\ref{condi11}). For instance, we select
		\begin{align}
		\hspace{-0.4cm}
		L_3(H|\theta^\circ_1)=0.4, \quad L_3(H|\theta^\circ_2)=0.3, \quad L_3(H|\theta^\circ_3)=0.8
		\end{align}
	   In this case, the belief for agent $3$ will be updated as:
		\begin{align}
		\bmu_{3,i}(\theta)&=a_{13}\bmu_{1,i}(\theta)+a_{23}\bmu_{2,i}(\theta)+a_{33}\bpsi_{3,i-1}(\theta) \nn \\
		{}&=a_{13}\bmu_{1,i}(\theta)+a_{23}\bmu_{2,i}(\theta)+ \nn  \\
		{}&\;a_{33}\frac{ L_3(\bm\xi_{3,i} | \theta) \bmu_{3,i-1} (\theta) } {\sum_{\theta ' \in \Theta} \bmu_{3,i-1} (\theta ') L_3(\bm\xi_{3,i} | \theta')}
		\label{withContri}
		\end{align}
		We see here how this equality is different from (\ref{noContri}), where the last term $\bm\psi_{3,i-1}(\theta)$ holds information about $\theta^\circ_3$ that contradicts with the information held in the other terms. We now show by contradiction that in this case, agent $3$ will not converge to a fixed distribution. Assume, to the contrary, that the beliefs of agent $3$ reach the following distribution: 
		\begin{align}
		\hspace{-0.3cm}
		\lim_{i\to\infty}\bmu_{3,i}(\theta^\circ_1)=b, \;\lim_{i\to\infty}\bmu_{3,i}(\theta^\circ_2)=c, \; \lim_{i\to\infty}\bmu_{3,i}(\theta^\circ_3)=d
		\end{align}
		for some fixed non-negative constants $b$, $c$ and $d$ satisfying
		\be 
		b+c+d=1
		\label{condiProba}
		\ee
		We know that, as $i\to\infty$, agents $1$ and $2$ approach their true states so that by evaluating (\ref{withContri}) at $\theta^\circ_1$ when $i\to\infty$, we get:
		\begin{align}
		b=a_{13}+\frac{a_{33}L_3(\bm\xi_{3,i}|\theta^\circ_1)b}{bL_3(\bm\xi_{3,i}|\theta^\circ_1)+cL_3(\bm\xi_{3,i}|\theta^\circ_2)+dL_3(\bm\xi_{3,i}|\theta^\circ_3)} 
		\label{contr1}
		\end{align}
		Evaluating (\ref{withContri}) at $\theta^\circ_2$ when $i\to\infty$:
		\begin{align}
		c=a_{23}+\frac{a_{33}L_3(\bm\xi_{3,i}|\theta^\circ_2)c}{bL_3(\bm\xi_{3,i}|\theta^\circ_1)+cL_3(\bm\xi_{3,i}|\theta^\circ_2)+dL_3(\bm\xi_{3,i}|\theta^\circ_3)} 
		\label{contr2}
		\end{align}
		Evaluating (\ref{withContri}) at $\theta^\circ_3$ when $i\to\infty$:
		\begin{align}
		d=\frac{a_{33}L_3(\bm\xi_{3,i}|\theta^\circ_3)d}{bL_3(\bm\xi_{3,i}|\theta^\circ_1)+cL_3(\bm\xi_{3,i}|\theta^\circ_2)+dL_3(\bm\xi_{3,i}|\theta^\circ_3)} 
		\label{contr3}
		\end{align}
		Then, from (\ref{contr3}), we have:
		\begin{align}
		d&=\frac{0.56d}{0.4b+0.3c+0.8d}, \; \;\text{if observation is H}  \\
		d&=\frac{0.14d}{0.6b+0.7c+0.2d}, \; \;\text{if observation is T} 
		\end{align}
		Then, either $d=0$ or
		\be 
		0.4b+0.3c+0.8d=0.56 \text{ and } 0.6b+0.7c+0.2d=0.14
		\label{first}
		\ee
		However, conditions (\ref{first}) contradict the fact that we must have
		\begin{align}
		\left(0.4b+0.3c+0.8d\right)+\left(0.6b+0.7c+0.2d\right)&=b+c+d \nn \\
		&\stackrel{(\ref{condiProba})}=1
		\end{align}
		We conclude that $d=0$. Thus, condition (\ref{condiProba}) reduces to:
		\be
		b+c=1
		\label{condiProba2} 
		\ee
		With regards to the values of $b$ and $c$, we know from (\ref{contr1}) that
		\begin{align}
		b&=0.1+\frac{0.28b}{0.4b+0.3c},\; \;\text{if observation is H} \label{contra22} \\
		b&=0.1+\frac{0.42b}{0.6b+0.7c}, \; \;\text{if observation is T} 
		\label{contra2}
		\end{align}
		That is, the scalars $b$ and $c$ must satisfy
		\be 
		\frac{0.28}{0.4b+0.3c}=\frac{0.42}{0.6b+0.7c}
		\ee 
		The denominators are related as follows:
		\be
		(0.4b+0.3c)+(0.6b+0.7c)=b+c\stackrel{(\ref{condiProba2})}=1 
		\ee
		Thus,
		\be 
		\frac{0.28}{0.4b+0.3c}=\frac{0.42}{1-(0.4b+0.3c)}
		\ee 
		This leads to \be 0.4b+0.3c=\frac{0.28}{0.28+0.42}=0.4\ee
		so that from (\ref{contra22}), we have
		\be 
		b=0.1+\frac{0.28}{0.4}b
		\ee
		Thus, $b=\frac{1}{3}$, and since $0.4b+0.3c=0.4$, then $c=\frac{8}{9}$. However, $b+c=\frac{11}{9}$, which contradicts (\ref{condiProba2}). We conclude that the beliefs of agent $3$ cannot reach a fixed distribution. This conclusion is illustrated in Fig. \ref{motiv}, which plots the evolution of beliefs of agent $3$ for all $\theta\in\Theta$. It is clear from the figure how the contradictory information conveyed by the influential agents and the private signals do not lead agent $3$ to approach a fixed belief. This also means that agents $1$ and $2$ cannot fully control agent $3$.
		\begin{figure}[h]
			\epsfxsize 9cm \epsfclipon
			\begin{center}
				\includegraphics[scale=0.4]{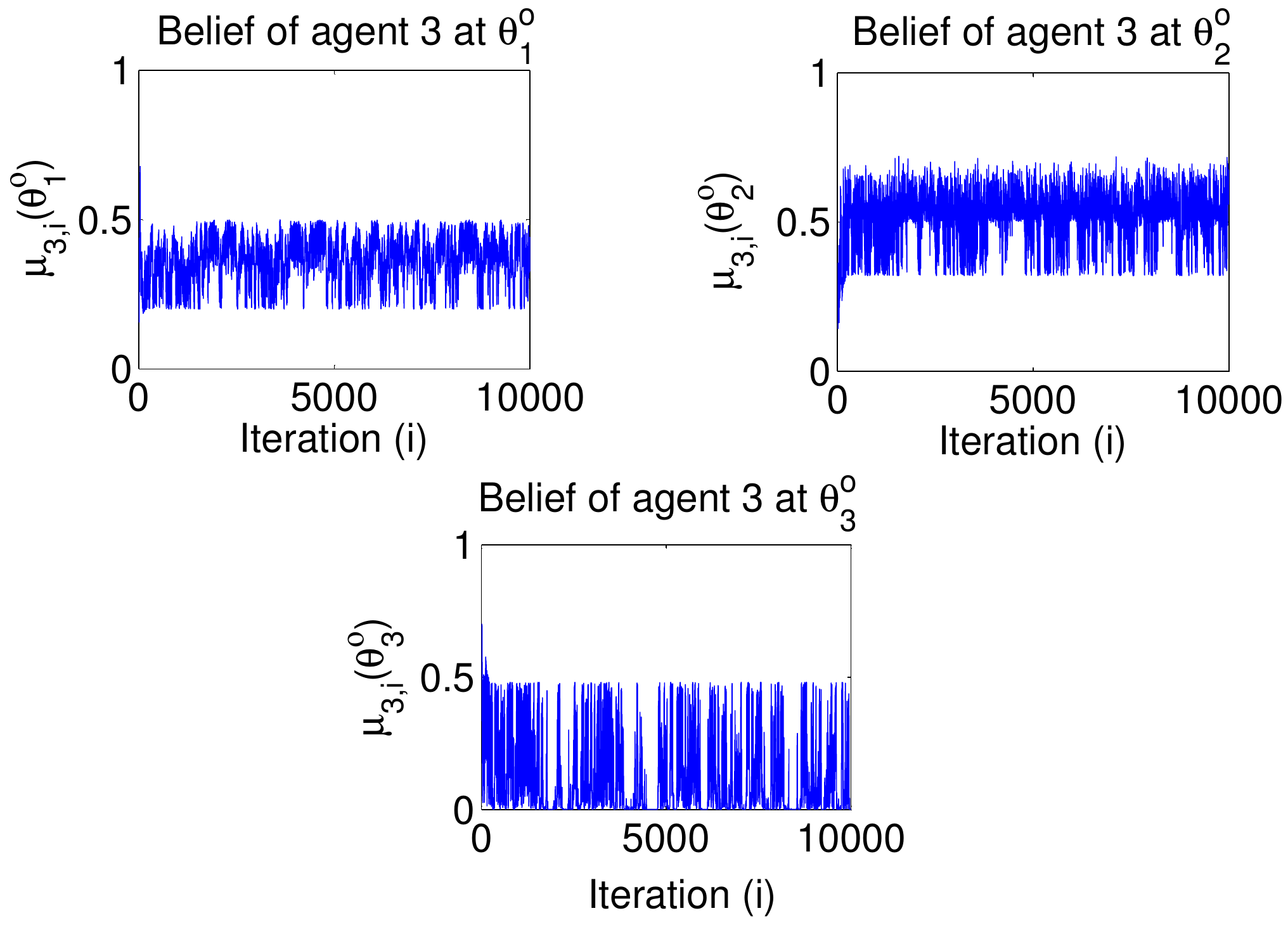}
				\caption{{\small Evolution of the beliefs of agent $3$ over time for the case in which condition (\ref{assumImp}) is not satisfied.}}\label{motiv}
			\end{center}
		\end{figure} 
		However, if agent $3$ decides to limit the contribution of its private signal on the update of its intermediate belief, will agents $1$ and $2$ be able to totally influence agent $3$? In other words, will the total influence scenario arise again even if assumption (\ref{assumImp}) is not satisfied? We show next that this is possible by incorporating an element of self-awareness into the learning process.
		\section{Diffusion Learning with Self-Awareness}
		We are therefore now motivated to modify the diffusion strategy (\ref{eqn:diffusion}) by incorporating a non-negative convex combination $\gamma_{k,i}$. This factor enables agents to assign more or less weight to their local information in comparison to the information received from their
		neighbors. Specifically, we modify
		(\ref{eqn:diffusion}) as follows:
		\be
		\left\{
		\begin{aligned}
			\bm\psi_{k,i}(\theta) &= \displaystyle \left(1-\gamma_{k,i}\right)\bmu_{k,i-1}(\theta) \\
			&\quad\quad + \gamma_{k,i}\frac{ \bmu_{k,i-1} (\theta) L_k(\bm\xi_{k,i} | \theta) } {\sum_{\theta ' \in \Theta} \bmu_{k,i-1} (\theta ') L_k(\bm\xi_{k,i} | \theta')} \\
			\bmu_{k,i}(\theta) &= \sum_{\ell \in \mathcal{N}_k } a_{\ell k} \bm\psi_{\ell,i} (\theta) \label{model2}
		\end{aligned}
		\right.
		\ee 
		where $\gamma_{k,i}\in[0,1]$ is a scalar variable. Observe that the intermediate belief  $\bpsi_{k,i}(\theta)$ of agent $k$ is now a combination of its prior belief, $\bmu_{k,i-1}(\theta)$, and the Bayesian update. The scalar $\gamma_{k,i}$ represents the amount of trust that agent $k$ gives to its private signal and how it is balancing this trust between the new observation and its own past belief. This weight can also model the lack of an observational signal at time $i$.
		
		Model (\ref{model2}) helps capture some elements of human behavior. For example, in an interactive social setting, a human agent may not be satisfied or convinced by an observation and prefers to give more weight to their prior belief based on accumulated experiences. This model was studied for single agents in\cite{epstein2010non, epstein2} and was motivated as a mechanism for self-control and temptation. The agent might observe a private signal at some time that can move this agent away from its current conviction. The agent can control this temptation by increasing the weight given to its prior belief or it can change its opinion by giving more weight to its Bayesian update, which is based on the private signal.
		
		We next analyze model (\ref{model2}) over weakly-connected graphs and establish two results. The first result is related to the sending agents
		and the second result is related to the receiving agents.
		
		\begin{lemma}[Correct Forecasting with Self-Awareness]
			Assume that $\lim\limits_{i\to\infty} \gamma_{k,i} \neq 0$ and the same conditions of Lemma 1. Then, self-aware sending agents develop correct forecasts of the incoming signals, namely, result (\ref{CorrectForecast}) continues to hold.
			\label{lemma2} 
		\end{lemma}
		\noindent {\em Proof}: See Appendix \ref{app.A}.
		\qd
		\begin{theorem}[Truth Learning by Self-Aware Sending Agents]
			\label{Theorem2}
			Under the same assumptions of Theorem \ref{theorem1}, self-aware sending agents learn the truth asymptotically and condition (\ref{truthLearning}) continues to hold.
		\end{theorem}
		\noindent {\em Proof}: The argument is similar to the proof given in \cite{zhao2012learning}.
		\qd 
		
		We therefore find that sending agents, whether self-aware or not, are always able to learn the truth. With regards to receiving agents, we now have the following conclusion. For each agent $k$ in a receiving sub-network $r$, we write  $\gamma_{k,i}=\tau_{k,i}\gamma_{\max}$, where $\gamma_{\max}$ are both positive scalars less than 1, and $\gamma_{\max} = \sup_{k,i} \gamma_{k,i}$.
		\begin{theorem} [Learning by Self-Aware Receiving Agents]
			\label{mainresultself}
			The beliefs of self-aware receiving agents are confined as follows:
			\begin{align}
			\limsup_{i\to\infty}\bmu_{\mathcal{R},i}(\theta) &\preceq W\tran\left(\lim_{i\to\infty}\bmu_{\mathcal{S},i}(\theta)\right)+ \gamma_{\max}C\one_{N_{gR}}  \\
			\liminf_{i\to\infty}\bmu_{\mathcal{R},i}(\theta) &\succeq  W\tran\left(\lim_{i\to\infty}\bmu_{\mathcal{S},i}(\theta)\right)- \gamma_{\max}C\one_{N_{gR}}  \label{resultfinal}
			\end{align}
			where $	C\define(I-T_{RR}\tran)^{-1}$ is an $N_{gR}\times N_{gR}$ matrix.
		\end{theorem}
		
		\noindent {\em Proof}: See Appendix \ref{App.C}.
		\qd \\
		
		\noindent This final result coincides with that of Theorem \ref{mainresult1}, but with an additional $O(\gamma_{\max})$ term. This means that if each receiving agent chooses the $\gamma-$coefficient to be small enough, then its belief converges to the same distribution (\ref{result}) of Theorem \ref{mainresult1}. When agent $k$ gives a small weight to its Bayesian update, it means that it is giving its current signal $\bm\xi_{k,i}$ a reduced role to play in affecting its belief formation at time $i$, and it is instead relying more heavily on its prior belief $\bmu_{k,i-1}(\theta)$ and on its communication with its neighbors. When agent $k$ continues to give less importance to any current signal it is receiving, its belief update will be mainly affected by its interaction with influential agents and its neighbors that are also under the influence of sending agents. Therefore, over time, these circumstances will help establish a leader-follower relationship in the network. In other words, the receiving sub-networks will be driven away from the truth and be under total indoctrination by the influential agents. 
		\section{Simulation Results}
		We illustrate the previous results for weakly-connected networks. We assume that the social network has $N=8$ agents interconnected as shown in Fig. \ref{network.label}, which corresponds to the following combination matrix:
		
		\begin{small}
			\begin{equation}
				A=\left[
				\begin{array}{ccccc|ccc}
					0.2		&	0.2 	     	&0.8		& 0 		&	0 		& 0		&0		&0	 \\
					0 .5   	&      0.4   		&0.1		& 0 		& 	0		& 0.2 	&0		 &0.4	 \\
					0.3 		& 	0.4		&0.1		&0		& 	0 		& 0.1		&0		&0	\\
					0          	& 	0  		&0		& 0.4 	& 	0.3 		& 0.3 	&0		&0	\\
					0          	& 	0  		&0		& 0.6 	& 	0.7 		& 0	 	&0		&0	\\
					\hline
					0          	& 	0  		&0		& 0 		& 	0		& 0.2	 	&0.3		&0.2	 \\
					0		&      0		&0		& 0		& 	0		& 0.1		&0.5		&0.3	 \\
					0          	& 	0  		&0		& 0 		& 	0		& 0.1		&0.2		&0.1	 \\
				\end{array}
				\right]
				\label{label.eq11}\end{equation}
		\end{small}
		
		\begin{figure}[h!]
			\centering
			\includegraphics[scale=0.45]{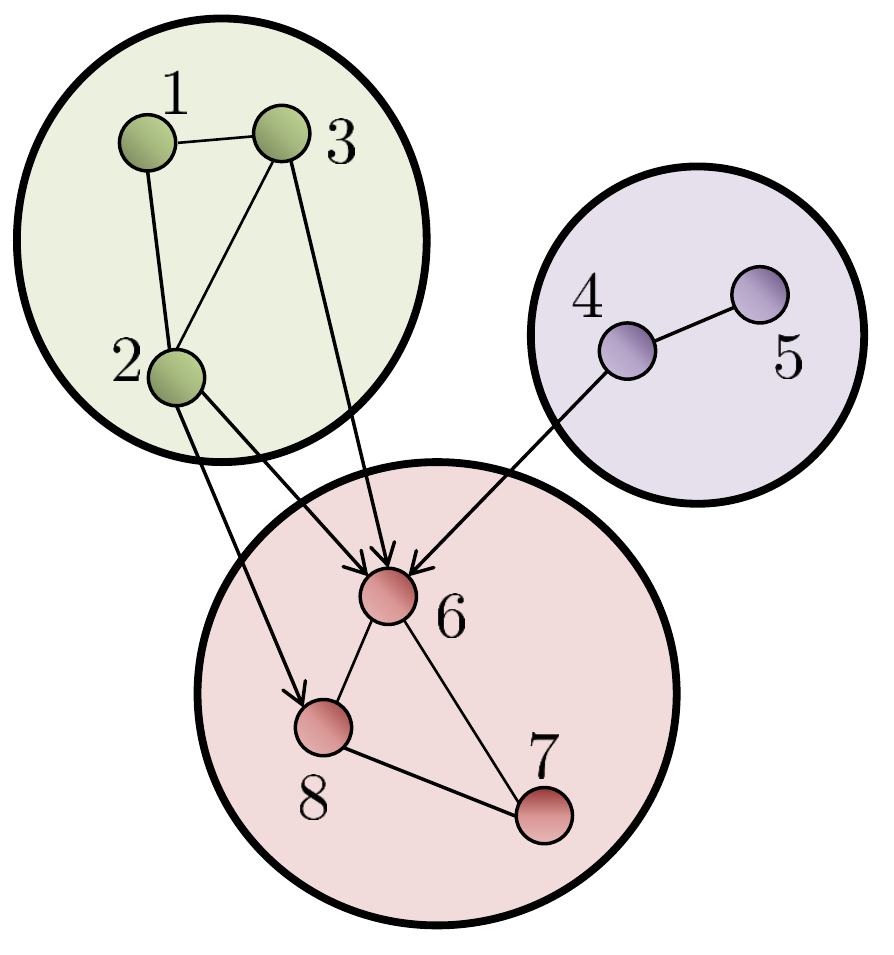} 
			\caption{{\small A weakly connected network consisting of three sub-networks and the corresponding combination policy (\ref{label.eq11}). }}
			\label{network.label}
		\end{figure}
		
		We assume that there are 3 possible events $\Theta=\{\theta_1^\circ,\theta_2^\circ,\theta_3^\circ\}$, where $\theta_1^\circ$ is the true event for the first sending sub-network, $\theta_2^\circ$ is the true event for the second sending sub-network, and $\theta_3^\circ$ is the true event for the receiving sub-network. We further assume that the observational signals of each agent $k$ are binary and belong to $Z_k=\{H,T\}$ where $H$ denotes head and $T$ denotes tail. We consider two cases. In the first case, we assume that agents update their belief according to the model described in (\ref{eqn:diffusion}) and that assumption (\ref{assumImp}) is met. In the second case, we assume that agents follow the second model described in (\ref{model2}) where assumption (\ref{assumImp}) is not met.
		\subsection{First Case}
		In this first case, the likelihood of the head signals for each agent $k$ is selected as the following  $3 \times 8$ matrix:
		\be
		L(H)
		=
		\ba{cccccccc}
		5/8& 3/4& 1/3& 7/8& 5/8& 1/3& 1/4& 5/8 \\
		5/8 & 1/4& 1/6& 7/8& 2/3& 1/3& 1/4& 5/8 \\ 
		1/4 & 3/4& 1/6& 1/3& 2/3& 1/3& 1/4& 5/8 \nn
		\ea
		\ee 
		where each $(j,k)$-th element of this matrix corresponds to $L_k(H/\theta_j)$, i.e., each column corresponds to one agent and each row to one network state. The likelihood of the tail signal is $L(T)=\one_{3\times 8}-L(H)$. We observe from $L(H)$ that assumption (\ref{assumImp}) is met here where for agent $k$ in the receiving sub-network ($k>5$) we have $L_k(\zeta_k|\theta_1^\circ)=L_k(\zeta_k|\theta_2^\circ)=L_k(\zeta_k|\theta_3^\circ)$ for both cases in which $\zeta_k$ is either head or tail.  Assumption (\ref{assumImp}) is met here because the true state of the first sending sub-network $\theta^\circ_1$ belongs to the indistinguishable set of any receiving agent $k$ in the receiving sub-network $3$, i.e., $L_k(\zeta_k|\theta^\circ_1)=L_k(\zeta_k|\theta^\circ_3)$, and the true state of the second sending sub-network $\theta^\circ_2$ belongs to the indistinguishable set of any receiving agent $k$, i.e., $L_k(\zeta_k|\theta^\circ_2)=L_k(\zeta_k|\theta^\circ_3)$, where $k=6,7,8$.
		We further assume that each agent starts at time $i=0$ with an initial belief that is uniform over $\Theta$ and then updates it over time according to the model described in (\ref{eqn:diffusion}). Then, we know from \cite{zhao2012learning} that $\lim_{i \to \infty} \bm\mu_{k,i}(\theta_1^\circ)=1$ for $k=1,2,3$ and  $\lim_{i \to \infty} \bm\mu_{k,i}(\theta_2^\circ)=1$ for $k=4,5$. Now for the agents of the receiving sub-network, we need first to compute:
		\begin{align}
		W\tran&=(I-T_{RR}\tran)^{-1}T_{SR}\tran \nn \\
		&=
		\ba{cccccc}
		0&         0.4045&     0.1489& &0.4466& 0\\
		0&         0.5267&     0.1183& &0.3550& 0\\
		0&         0.7099&     0.0725& &0.2176& 0
		\ea
		\end{align}
		The first row of $W\tran$ corresponds to agent $6$, the second row to agent $7$ and the third row to agent $8$. Now each row is partitioned into two blocks: the first block is of length $N_1=3$ that corresponds to sub-network $1$ of true state $\theta^\circ_1$ and the second block is of length $N_2=2$ that corresponds to sub-network $2$ of true state $\theta^\circ_2$.  Then, according to Theorem \ref{mainresult1}, we can compute the belief at $\theta_1^\circ$ for each receiving agent at steady state, by taking the first block in the agent's corresponding row and summing its elements:
		$$\lim_{i\to\infty}\bm \mu_{k,i}(\theta_1^\circ)=\left\{
		\begin{aligned}
		0+0.4045+0.1489=0.5534,\ &\quad k=6\\
		0+0.5267+0.1183=0.6450,\ &\quad  k=7\\
		0+0.7099+0.0725=0.7824,\ &\quad  k=8
		\end{aligned}\right.$$
		Likewise, we can compute the belief at $\theta_2^\circ$ for each receiving agent at steady state, by taking the second block in the agent's corresponding row and summing its elements:
		$$\lim_{i\to\infty}\bm \mu_{k,i}(\theta_2^\circ)=
		\left\{
		\begin{aligned}
		0.4466+0=0.4466,\ &\quad k=6\\
		0.3550+0=0.3550,\ &\quad  k=7\\
		0.2176+0=0.2176,\ &\quad  k=8
		\end{aligned}\right.$$
		%
		%
		
		We run this example for 7000 time iterations. We assigned to each agent an initial belief that is uniform over $\{\theta_1^\circ,\theta_2^\circ,\theta_3^\circ\}$.
		Figures \ref{figThetaExp1.label} shows the evolution of $\bm\mu_{k,i}(\theta_1^\circ)$ and $\bm\mu_{k,i}(\theta_2^\circ)$ of agents in the receiving sub-network $(k=6,7,8)$. These figures show the convergence of the beliefs of the agents in the receiving sub-networks to the same probability distribution already computed according to the results of Theorem \ref{mainresult1}.  Figure \ref{pmf.label} shows this limiting distribution over $\Theta$ for all receiving agents.
	    \begin{figure}[h]
	    	\epsfxsize 9cm \epsfclipon
	    	\begin{center}
	    		\epsffile{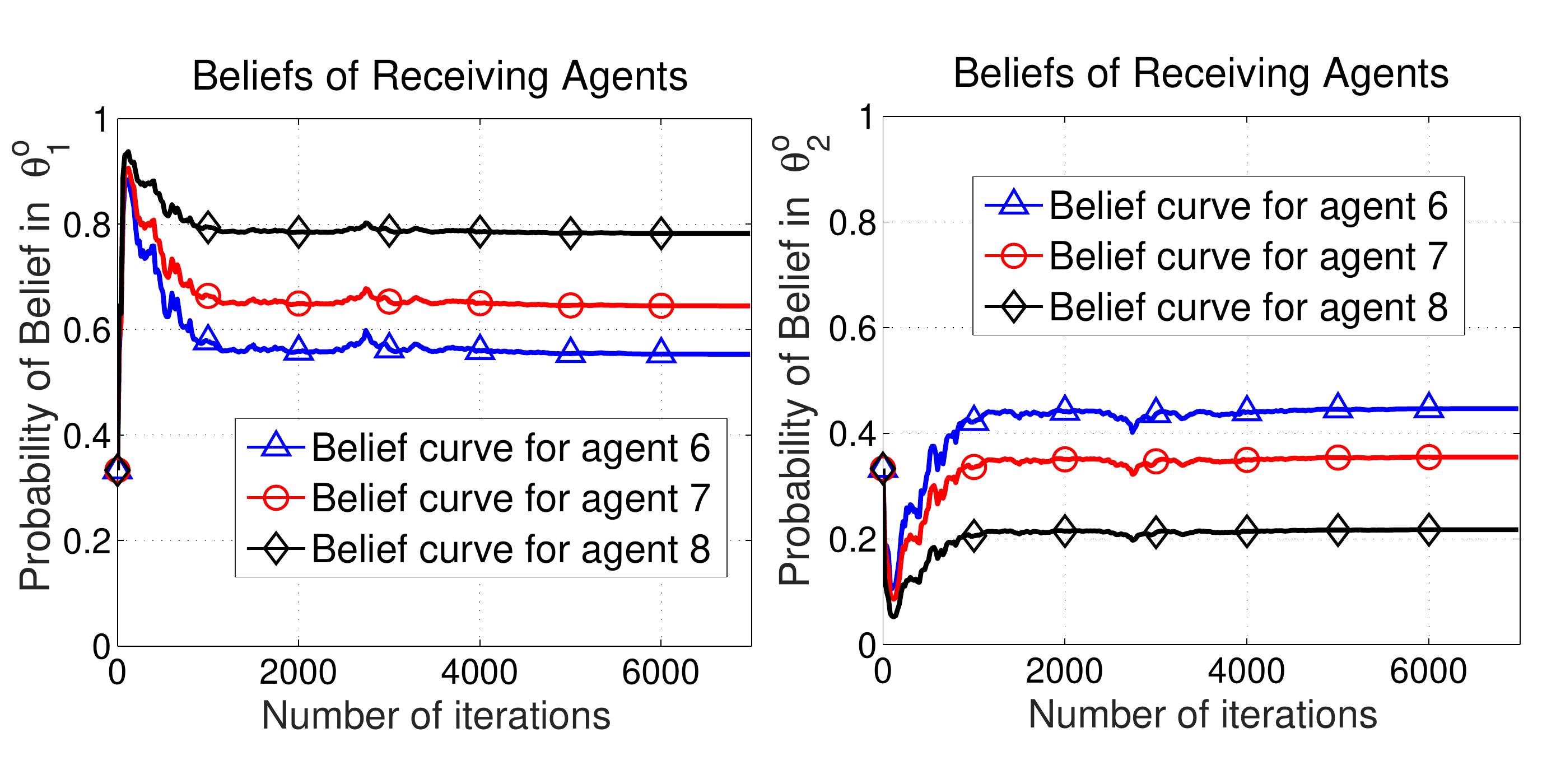} \caption{{\small Evolution of agent $k$ belief over time for $k=6,7,8$} in the first case.\vspace{-0.8cm}}\label{figThetaExp1.label}
	    	\end{center}
	    \end{figure}

		\begin{figure}[h]
			\epsfxsize 8cm \epsfclipon
			\begin{center}
				\epsffile{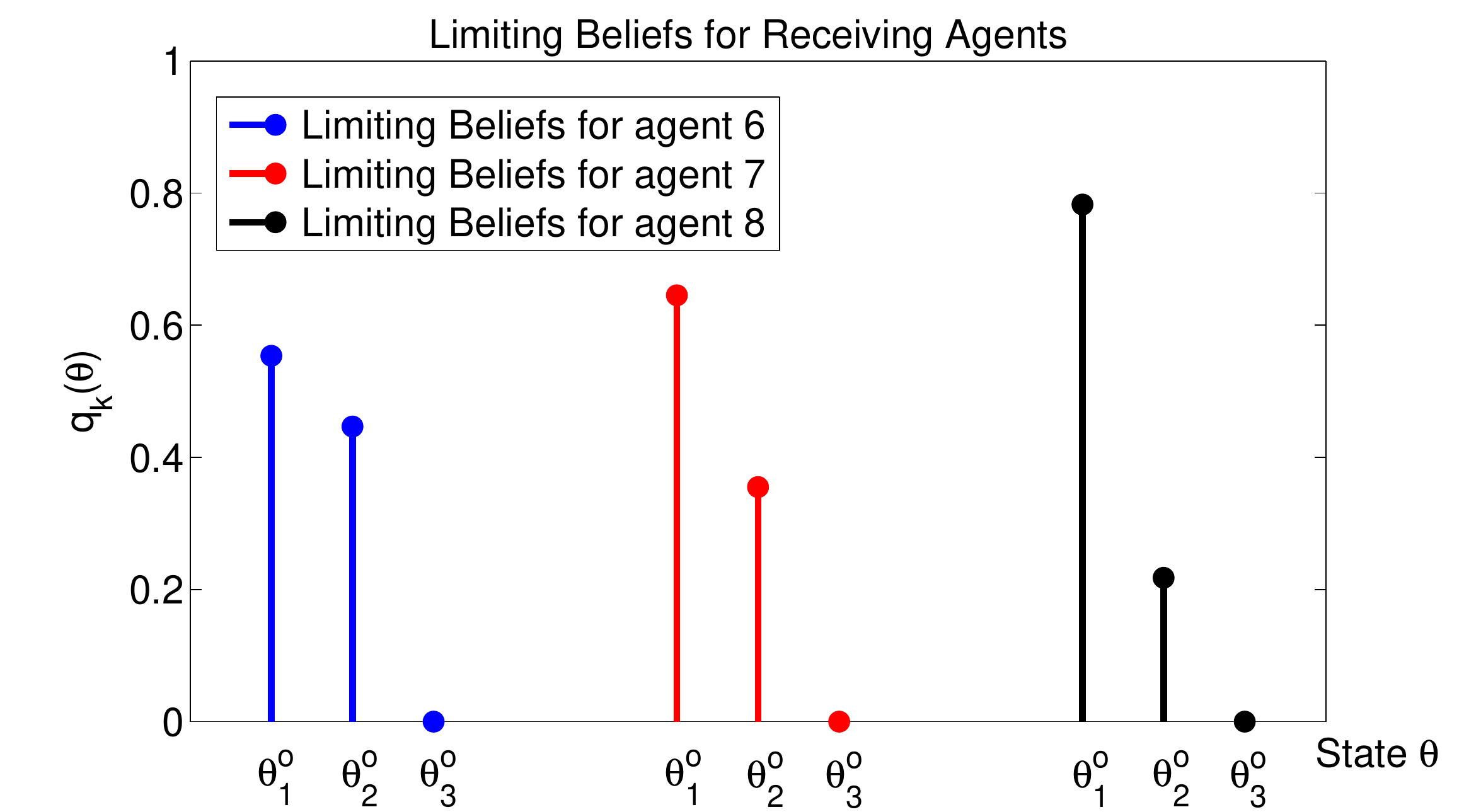} \caption{{\small Limiting distribution of agent $k$, $q_{k}(\theta)$, over $\Theta$ for $k=6,7,8$}}\label{pmf.label}
			\end{center}
		\end{figure}
		\subsection{Second Case}
		We now assume that the likelihood of the head signals for each agent $k$ is selected as the following  $3 \times 8$ matrix:
		 \be
		  L(H)
		  =
		  \ba{cccccccc}
		  5/8& 3/4& 1/3& 7/8& 5/8& 1/2& 2/3& 3/8 \\
		  5/8 & 1/4& 1/6& 7/8& 2/3& 1/3& 3/5& 5/7 \\ 
		  1/4 & 3/4& 1/6& 1/3& 2/3& 2/5& 1/4& 1/3 \nn
		  \ea
		  \ee 
		We observe now from $L(H)$ that assumption (\ref{assumImp}) is not met here where for agent $k$ in the receiving sub-network ($k>5$) we have $L_k(\zeta_k|\theta_1^\circ)\neq L_k(\zeta_k|\theta_2^\circ)\neq L_k(\zeta_k|\theta_3^\circ)$ for both cases in which $\zeta_k$ is either head or tail.  Assumption (\ref{assumImp}) is not met here because $\theta^\circ_1$ does not belong to the indistinguishable set of any receiving agent $k$ in the receiving sub-network $3$, i.e., $L_k(\zeta_k|\theta^\circ_1)\neq L_k(\zeta_k|\theta^\circ_3)$, and $\theta^\circ_2$ does not belong to the indistinguishable set of any receiving agent $k$, i.e., $L_k(\zeta_k|\theta^\circ_2)\neq L_k(\zeta_k|\theta^\circ_3)$, where $k=6,7,8$. We further assume that agents now update their beliefs according to the model described in (\ref{model2}). We choose $\gamma_{k,i}=0.4$ for $k=1,2,3$ (agents of the first sending sub-network) at any $i$, $\gamma_{k,i}=0.5$ for $k=4,5$ (agents of the second sending sub-network) at any $i$ and $\gamma_{k,i}=0.1$ for $k=6,7,8$ (agents of the receiving sub-network) at any $i$. We also assume that each agent starts at time $i=0$ with an initial belief that is uniform over $\Theta$. Then, we know from Theorem \ref{Theorem2} that $\lim_{i \to \infty} \bm\mu_{k,i}(\theta_1^\circ)=1$ for $k=1,2,3$ and  $\lim_{i \to \infty} \bm\mu_{k,i}(\theta_2^\circ)=1$ for $k=4,5$. Figure \ref{figThetaExp2.label} shows the evolution of $\bm\mu_{k,i}(\theta_1^\circ)$ and $\bm\mu_{k,i}(\theta_2^\circ)$ of agents in the receiving sub-network $(k=6,7,8)$. These figures show how the beliefs of the receiving agents are confined around the probability distribution already computed in the previous case.
		    \begin{figure}[h]
		    	\epsfxsize 9cm \epsfclipon
		    	\begin{center}
		    		\epsffile{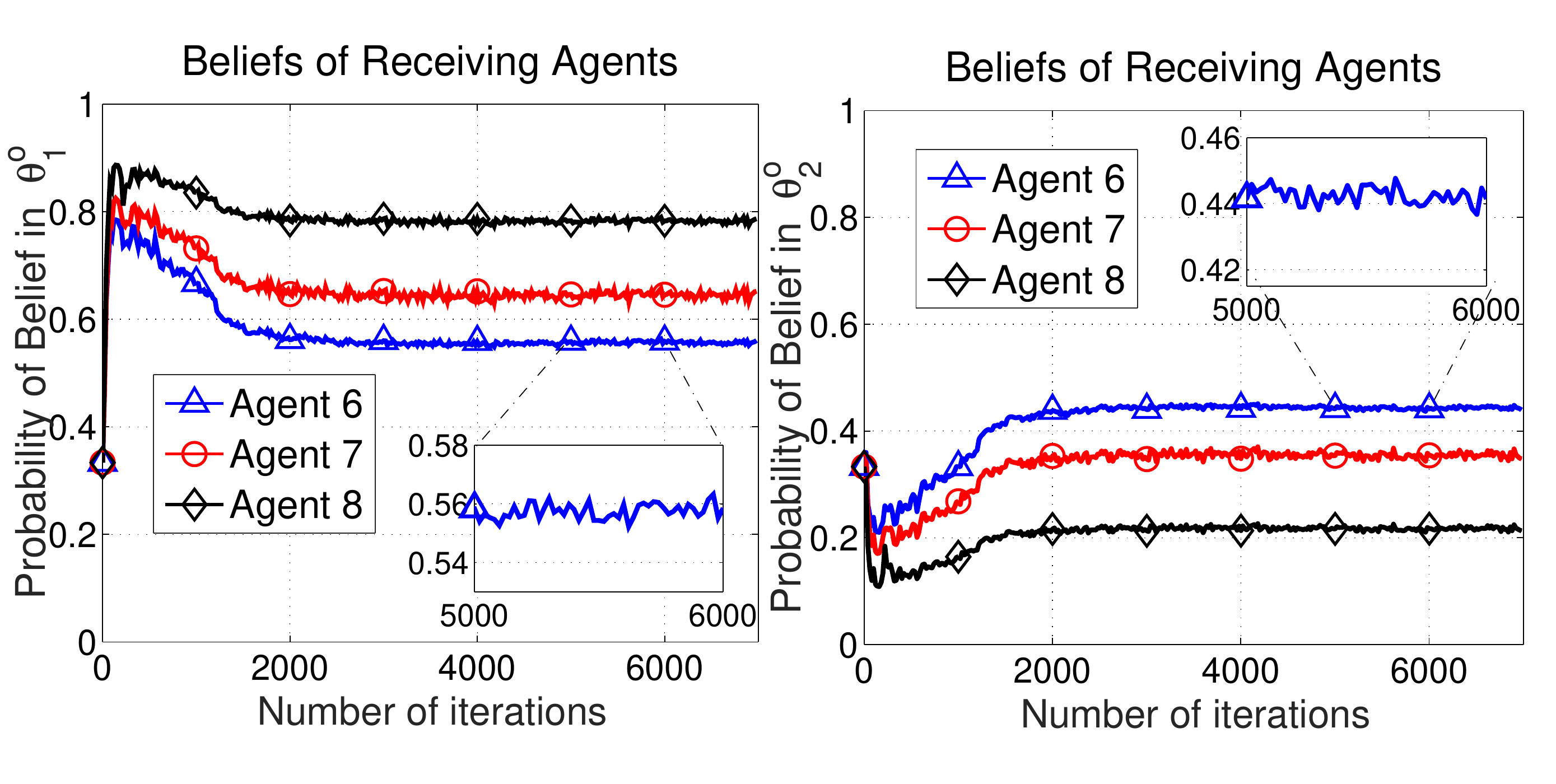} \caption{{\small Evolution of agent $k$ belief over time for $k=6,7,8$} in the second case.}\label{figThetaExp2.label}
		    	\end{center}
		    \end{figure}
		\section{Conclusion}
		In this article, we studied diffusion social learning over weakly-connected networks. We examined the circumstances under which receiving agents come under the total influence of sending agents. This total influence is reflected by forcing the receiving agents to focus their beliefs on the set of true states for the sending sub-networks. We determined for each receiving agent what the exact probability distribution is in steady-state. We also illustrated the results with examples. Future work will focus on how the network can be designed so that receiving agents adopt \emph{specific} limiting beliefs, and how receiving agents can detect the external influence and limit it.
  
 \appendices
\section{Proof of Lemma \ref{lemma:vanishing}} \label{App.B}
The proof is based on showing first that for any receiving agent $k$, it holds that
\be \lim\limits_{i\to\infty}\sum_{\theta\in\Theta^\bullet}\bmu_{k,i}(\theta)=1
\ee
From this result, we will conclude that $\lim\limits_{i\to\infty}\bmu_{k,i}(\theta)=0$ for all $\theta\in\bar\Theta^\bullet$. To examine the evolution of agents' beliefs toward $\Theta^\bullet$, we associate with each agent $k$ the following regret function:
 \begin{align}
    Q^W(\bm\mu_{k,i})\define -\log\left( \sum_{\theta\in\Theta^\bullet}\bm\mu_{k,i}(\theta) \right)
    \label{regretWeak}
    \end{align}
 We view $\bm\mu_{k,i}(\theta)$ as a stochastic process that depends on the sequence of random observations $\{\bm{\xi}_{k,j}\}$ over all $k$ and for all $j\leq i$. Therefore, we shall  examine agent $k$'s individual performance by taking the expectation of $Q^W(\bm\mu_{k,i})$ over these observations. More specifically, we define agent $k$'s risk at time $i$ as
    \begin{align}
    \hspace*{-0.5cm}
    J^W(\bmu_{k,i})\define \Ex_{\mathcal{F}_i} Q^W(\bmu_{k,i})
    \label{riskWe}
    \end{align}
   where $\mathcal{F}_i$ denotes the of sequence $\{\bm{\xi}_{k,j}\}$ over all $k$ and for all $j\leq i$.
    
   \vspace{0.4cm}  
   \noindent {\bf Proof of Lemma \ref{lemma:vanishing}}. We start with agent $k$'s risk at time $i$ defined in (\ref{riskWe}), where $k>N_{gS}$. Recall that $N=N_{gS}+N_{gR}$ represents the total number of agents in the whole network:
    	\begin{align}
    	J^W({\bmu}_{k,i})\hspace{-1cm}& \nn\\
    	&= -  \Exf \log \left( \sum_{\theta\in\Theta^\bullet} \bmu_{k,i}(\theta)\right) \nn\\
    	&\stackrel{(\ref{eqn:diffusion})}{=}	 -\Exf \displaystyle \log \left[ \sum_{\ell=1}^{N}\sum_{\theta\in\Theta^\bullet}a_{\ell k} \boldsymbol\psi_{\ell,i}(\theta)\right] \nn \\
   		&\stackrel{(\ref{eqn:bayesianupdate})}{=} -\Exf \displaystyle \log \left[ \sum_{\ell=1}^{N_{gS}}a_{\ell k}  \sum_{\theta\in\Theta^\bullet}\boldsymbol\psi_{\ell,i}(\theta) \right. \nn \\
   		&\qquad+ \left.   \sum_{\ell = N_{gS}+1}^{N}a_{\ell k}\frac{\sum_{\theta\in\Theta^\bullet} \bmu_{\ell,i-1}(\theta) L_{\ell}(\bxi_{\ell,i}|\theta)  } { \m_{\ell,i-1} (\bxi_{\ell,i})} \right] \nn \\
    	&\stackrel{(a)}{=} -\Exf \displaystyle \log \left[  \sum_{\ell=1}^{N_{gS}}a_{\ell k} \sum_{\theta\in\Theta^\bullet} \boldsymbol\psi_{\ell,i}(\theta) \right.\nn \\ &\qquad+\left.\sum_{r=S+1}^{S+R} \sum_{\ell \in \mathcal{I}_r}a_{\ell k}\frac{ \sum_{\theta\in\Theta^\bullet} \bmu_{\ell,i-1}(\theta) L_{\ell}(\bxi_{\ell,i}|\theta)  } { \m_{\ell,i-1}(\bxi_{\ell,i}) } \right] \nn \\
    	&\stackrel{(b)}{=} -\Exf \displaystyle \log \left[  \sum_{\ell=1}^{N_{gS}}a_{\ell k} \sum_{\theta\in\Theta^\bullet} \boldsymbol\psi_{\ell,i}(\theta) \right.\nn \\ &\qquad+\left.\sum_{r=S+1}^{S+R} \sum_{\ell \in \mathcal{I}_r}a_{\ell k}\frac{ \sum_{\theta\in\Theta^\bullet} \bmu_{\ell,i-1}(\theta) L_{\ell}(\bxi_{\ell,i}|\theta^\circ_r)  } { \m_{\ell,i-1}(\bxi_{\ell,i}) } \right] \nn \\
    	&\stackrel{(c)}{\leq} -\Exf \displaystyle  \left[  \sum_{\ell=1}^{N_{gS}}a_{\ell k} \log \left(\sum_{\theta\in\Theta^\bullet} \boldsymbol\psi_{\ell,i}(\theta)\right)+\sum_{r=S+1}^{S+R} \sum_{\ell \in \mathcal{I}_r}a_{\ell k}\right.\nn \\
    	&\qquad\qquad\qquad\left.\log\frac{ \left(\sum_{\theta\in\Theta^\bullet} \bmu_{\ell,i-1}(\theta)\right) L_{\ell}(\bxi_{\ell,i}|\theta^\circ_r)  } { \m_{\ell,i-1}(\bxi_{\ell,i}) } \right] \nn \\
    	&\stackrel{}{=} -\sum_{\ell=1}^{N_{gS}}a_{\ell k} \Exf \displaystyle \log \left( \sum_{\theta\in\Theta^\bullet} \boldsymbol\psi_{\ell,i}(\theta) \right)\nn \\ 	&  \quad-\displaystyle \sum_{r=S+1}^{S+R} \sum_{\ell \in \mathcal{I}_r}a_{\ell k}  \Exf \log \left( { \sum_{\theta\in\Theta^\bullet} \bmu_{\ell,i-1}(\theta)  }  \right) \nn \\
    	&  \quad -\displaystyle \sum_{r=S+1}^{S+R} \sum_{\ell \in \mathcal{I}_r}a_{\ell k}  \Exf \log \left( \frac{  L_{\ell}(\bxi_{\ell,i}|\theta^{\circ}_r)  } { \m_{\ell,i-1}(\bxi_{\ell,i}) } \right)  \nn\\
    	&\stackrel{(\ref{riskWe})}{=} \displaystyle \sum_{\ell =1}^{N_{gS}}a_{\ell k} J^W({\bpsi}_{\ell,i})+\sum_{r=S+1}^{S+R} \sum_{\ell \in \mathcal{I}_r}a_{\ell k} J^W({\bmu}_{\ell,i-1})  \nn\\
    	& \quad-\displaystyle \sum_{r=S+1}^{S+R} \sum_{\ell \in \mathcal{I}_r}a_{\ell k}  \Exf \log \left(\frac{  L_{\ell}(\bxi_{\ell,i}|\theta^{\circ}_r)  } { \m_{\ell,i-1}(\bxi_{\ell,i}) } \right)  \nn\\
   		&\stackrel{(d)}{=} \displaystyle \sum_{\ell =1}^{N_{gS}}a_{\ell k} J^W({\bpsi}_{\ell,i})+\sum_{r=S+1}^{S+R} \sum_{\ell \in \mathcal{I}_r}a_{\ell k} J^W({\bmu}_{\ell,i-1}) \nn\\
   		&  \qquad-\displaystyle  \Ex _{\mathcal{F}_{i-1}}\left(\sum_{r=S+1}^{S+R} \sum_{\ell \in \mathcal{I}_r}a_{\ell k} \right. \nn\\
   		&  \qquad\qquad\left.\Ex _{\bm\xi{\ell,i}} \left(\log \left(\frac{  L_{\ell}(\bxi_{\ell,i}|\theta^{\circ}_r)  } { \m_{\ell,i-1}(\bxi_{\ell,i}) } \right)|\mathcal{F}_{i-1}\right)\right)  \nn\\
    	&\stackrel{(e)}{\leq} \displaystyle \sum_{r=S+1}^{S+R} \sum_{\ell \in \mathcal{I}_r}a_{\ell k} J^W({\bmu}_{\ell,i-1}) + \sum_{\ell =1}^{ N_{gS}}a_{\ell k} J^W({\bpsi}_{\ell,i}) 
    	\label{longEqua}
    	\end{align}
    	where 
    	\begin{itemize} 
    		\item in the third equality, we only expanded the second term that corresponds to receiving agents in order to study its behavior. We did not do the same thing with the first term because it corresponds to sending agents and we already know how that $\psi_{\ell,i}(\theta)$ will converge with time for any sending agent $\ell$, as later shown in (\ref{sendingPsi}).
    		\item in step $(a)$, we split the second summation corresponding to receiving agents into $R$ groups, with each group corresponding to one receiving sub-network. Moreover, the symbol $\mathcal{I}_r$ denotes the set of indexes of agents that belong to receiving sub-network $r$;
    		\item in step $(b)$, we replaced $L_\ell(\bm\xi_{\ell,i}|\theta)$ by $L_\ell(\bm\xi_{\ell,i}|\theta^\circ_r)$. This follows from assumption (\ref{assumImp}): for any $\theta$ that is in $\Theta^\bullet$, $L_\ell(\zeta_{\ell}|\theta)=L_\ell(\zeta_{\ell}|\theta^\circ_r)$, for any $\zeta_{\ell}\in Z_{\ell}$;	
    		\item in step $(c)$, we applied the convexity property of $-\log(.)$ since the elements $\{a_{\ell k}\}$ form a convex combination for each agent $k$;
    		\item in step $(d)$, we applied the conditional expectation property $\left(\Ex_X[g(X)]=\Ex_Y[\Ex_{X|Y}[g(X)|Y]]\right)$ as follows: 
    		\be
    		\hspace*{-4.5cm} \Ex _{\mathcal{F}_{i}} \log \left( \frac{  L_{\ell}(\bxi_{\ell,i}|\theta^{\circ})  } { \m_{\ell,i-1}(\bxi_{\ell,i}) } \right) \nn 
    		\ee
    		\vspace*{-0.5cm}
    		\bq
    		&=&\Ex_{\mathcal{F}_{i-1}}\left(\Ex _{\mathcal{F}_{i}|\mathcal{F}_{i-1}} \left(\log \frac{  L_{\ell}(\bxi_{\ell,i}|\theta^{\circ})  } {\bm m_{\ell,i-1}(\bm\xi_{\ell,i}) } |\mathcal{F}_{i-1}\right)\right)  \nn \\
    		&=&\Ex_{\mathcal{F}_{i-1}}\left(\Ex _{\bm\xi_{\ell,i}} \left(\log \frac{  L_{\ell}(\bxi_{\ell,i}|\theta^{\circ})  } {\bm m_{\ell,i-1}(\bm\xi_{\ell,i}) } |\mathcal{F}_{i-1}\right)\right)
    		\label{condProperty}
    		\eq 
    		\item in step $(e)$, we replaced the previous expression in $(d)$ by an upper bound using the non-negativity of the KL-divergence from $L_\ell(.|\theta^\circ_r)$ to $\bm\m_{\ell,i-1}(.)$ \cite{Cover}.
    	\end{itemize}  
    	\vspace{-0.24cm}	
To continue with the argument we collect the risk values of $S-$agents and $R-$agents into two vectors as follows:
 \begin{align}
 	\hspace{-0.1cm}J^W(\bpsi_{\mathcal{S},i})\define&\;
 	{\rm col}\left\{
 	J^W(\bpsi_{1,i} ),
 	\hdots,
 	J^W(\bpsi_{N_{gS},i} ) 
 	\right\}\label{JR}  \\
 	\hspace{-0.1cm}J^W(\bmu_{\mathcal{R},i})\define&\;
 	{\rm col} \left\{
 	J^W(\bmu_{N_{gS}+1,i} ),
 	\hdots,
 	J^W(\bmu_{N,i} ) 
 	\right\}
 	\label{JS}
 \end{align}
Then, from (\ref{longEqua}), we write the vector inequality:
\be
J^W(\bmu_{\mathcal{R},i}) \preceq T_{RR}\tran J^W(\bmu_{\mathcal{R},i-1})+T_{SR}\tran J^W(\bpsi_{\mathcal{S},i})
\ee
We now establish the convergence of this inequality. We first consider the term $J^W(\bpsi_{\mathcal{S},i})$. We know that agents in the sending sub-networks can learn the truth if the assumptions mentioned in Lemma  \ref{lemma1} and Theorem \ref{theorem1} are met. One of the assumptions is that at least one agent in each strongly-connected sub-network $s$ starts with a non-zero prior belief at $\theta^\circ_s$. Let us denote this agent by $\ell_o$. As shown in \cite{zhao2012learning}, this condition guarantees that for large enough $i$, $\bm{\mu}_{k,i}(\theta_s^\circ)>0$ for all $k$ in this sub-network.  Accordingly, it also holds that for large enough $i$ agents in this sub-network will have nonzero intermediate beliefs at $\theta_s^\circ$, i.e., $\bm{\psi}_{k,i}(\theta_s^\circ)>0$. This implies that $J^W(\bpsi_{\mathcal{S},i})\succeq0$ and $T_{SR}\tran J^W(\bpsi_{\mathcal{S},i})\succeq0$ for large enough $i$ since the elements of $T_{SR}$ are all non-negative. Let us now consider agent $k'$ of a receiving sub-network $r$, which has agent $\ell'$ from sending sub-network $s$ in its neighborhood. After large enough $i$, 
\be 
\bmu_{k',i}(\theta^\circ_s)=\sum_{\ell \in \mathcal{N}_{k'}} a_{\ell k'} \bpsi_{\ell}(\theta^\circ_s)\geq a_{\ell' k'} \bpsi_{\ell'}(\theta^\circ_s) >0
\ee 
Then, in the next time step, all agents of sub-network $r$ that have agent $k'$ in their neighborhood will have non-zero belief at $\theta^\circ_s$. Since the received sub-network $r$ is connected, it follows that after large enough $i$,
\be 
\bmu_{k,i}(\theta^\circ_s)>0 \implies  \sum_{\theta \in \Theta^\bullet}\bmu_{k,i}(\theta)>0
\ee
for all agents $k$ that belong to sub-network $r$. We employ the same argument for all other receiving sub-networks. Therefore, $\sum_{\theta \in \Theta^\bullet}\bmu_{k,i}(\theta)>0$ for any $k>N_{gR}$ so that $J^W(\bmu_{\mathcal{R},i})\succeq0$ for large enough $i$.
Thus,
\be 
0\preceq J^W(\bmu_{\mathcal{R},i}) \preceq T_{RR}\tran J^W(\bmu_{\mathcal{R},i-1})+T_{SR}\tran J^W(\bpsi_{\mathcal{S},i})
\ee
Furthermore, any agent $k$ in any sending sub-network $s$ can learn asymptotically its own true state, so that $\lim_{i\to\infty}\bmu_{k,i}(\theta^\circ_s)\aseq1$ implies
\begin{align}
 \lim_{i\to\infty}\bpsi_{k,i}(\theta^\circ_s)&=\lim_{i\to\infty} \frac{\bmu_{k,i}(\theta^\circ_s)L_k(\bm\xi_{k,i}|\theta^\circ_s)}{\sum_{\theta\in\Theta}\bmu_{k,i}(\theta)L_k(\bm\xi_{k,i}|\theta)} \nn \\
 &=\lim_{i\to\infty} \frac{L_k(\bm\xi_{k,i}|\theta^\circ_s)}{L_k(\bm\xi_{k,i}|\theta^\circ_s)}\aseq 1 \label{sendingPsi}
\end{align}
The denominator in the second equality follows from the fact that $\lim_{i\to\infty}\bmu_{k,i}(\theta^\circ_s)\aseq1$ for any agent $k$ of sending sub-network $s$. It follows that $\lim\limits_{i\to\infty}\sum_{\theta\in\Theta^\bullet}\bpsi_{k,i}(\theta)\aseq1$ for any $k\leq{N_{gS}}$.
Therefore, $\lim\limits_{i\to\infty}J^W(\bpsi_{S,i})=0$. Moreover, since $\rho(T_{RR})<1$ \cite{ying2014information}, we conclude that
$
\lim_{i\to\infty}J^W(\bmu_{\mathcal{R},i})=0
$
which implies that 
    	\be
    	\lim_{i\to\infty} J^W({\bmu}_{k,i})  = 0 , \quad \forall \; k>N_{gS} \label{eq.costallJzero}
    	\ee
  As previously discussed after large enough $i$, $\sum_{\theta\in\Theta^\bullet}\bmu_{k,i}(\theta)>0$ so that $-\log\left(\sum_{\theta\in\Theta^\bullet}\bmu_{k,i}(\theta)\right) \geq 0$. Using the definition of $J^W(\bmu_{k,i})$ in (\ref{riskWe}), it holds that $J^W(\bmu_{k,i})$ represents the expectation over $\mathcal{F}_{i}$ of non-negative quantities. Hence, result (\ref{eq.costallJzero}) implies
  \be
  \hspace*{-0.3cm}\lim_{i\to\infty} -\log\left(\sum_{\theta\in\Theta^\bullet}\bmu_{k,i}(\theta)\right)=0 
 \text{ ;} \lim_{i\to\infty}  \sum_{\theta\in\Theta^\bullet}\bmu_{k,i}(\theta)\aseq1
  \ee \qd
  \vspace{-0.5cm}
   \section{Proof of Lemma \ref{lemma:correct}}
   \label{App.B1}Assume agent $k$ belongs to sub-network $r$ and $\zeta_k\in\Z_k$:
   	\begin{align}
   	\lim_{i\to\infty}\m_{k,i}(\zeta_k)&= \lim_{i\to\infty} \sum_{\theta \in \Theta} \bmu_{k,i} (\theta) L_k(\zeta_{k} | \theta) \nn \\
   	{}&\stackrel{(a)}=\lim_{i\to\infty}\sum_{\theta\in\Theta^\bullet} \bmu_{k,i} (\theta) L_k(\zeta_{k} | \theta) \nn \\
   	{}&\stackrel{(b)}=\left(\lim_{i\to\infty} \sum_{\theta\in\Theta^\bullet} \bmu_{k,i} (\theta)\right) L_k(\zeta_{k} | \theta_r^\circ) \nn \\
   	{}&\aseq L_{k} ({\zeta}_k|\theta^{\circ}_r) \label{eq.m=l}
   	\end{align}
    where step $(a)$ follows from the result of Lemma \ref{lemma:vanishing} and step $(b)$ follows from assumption (\ref{assumImp}).
    \qd 
    \section{Proof of Theorem \ref{mainresult1}}\label{App.B2} The intermediate belief of any agent $k$ is given by:
    \be
    \bm\psi_{k,i}(\theta) = \frac{\bmu_{k,i-1}(\theta)L_k(\bm\xi_{k,i}|\theta)}{\bm m_{k,i-1}(\bm\xi_{k,i})}
    \ee
    Let us assume that agent $k$ belongs to receiving sub-network $r$.  Using Lemma \ref{lemma:correct}, we have for any  $\theta\in\Theta^\bullet$:
    \begin{align} 
    \lim_{i\to\infty}\bm\psi_{k,i}(\theta)=\lim_{i\to\infty}\frac{\bmu_{k,i-1}(\theta)L_k(\bm\xi_{k,i}|\theta)}{\bm m_{k,i-1}(\bm\xi_{k,i})} =\lim_{i\to\infty}\bmu_{k,i-1}(\theta)
    \end{align}
    We can establish the same property for any agent in a sending sub-network because (\ref{eq.m=l}) was already proven for sending agents in \cite{zhao2012learning}. It follows that, for any agent $k$,
    \be
    \lim_{i \to \infty} \bmu_{k,i} (\theta)  = \lim_{i \to \infty} \sum_{\ell \in \mathcal{N}_k} a_{\ell k}\bmu_{k,i-1} (\theta) 
    \label{similarResult}
    \ee
    for any $\theta \in \Theta^\bullet$. We defined the vectors $\bmu_{\mathcal{S},i}(\theta)$ in (\ref{muS}) and $\bmu_{\mathcal{R},i}(\theta)$ in (\ref{muR}). Then,
    \begin{align} 
   \lim_{i\to\infty} 
    \ba {c}
    \bmu_{\mathcal{S},i} (\theta) \\
    \bmu_{\mathcal{R},i}(\theta) 
    \ea =
    A\tran \left(\lim_{i\to\infty}
    \ba {c}
    \bmu_{\mathcal{S},i-1} (\theta) \\
    \bmu_{\mathcal{R},i-1}(\theta) 
    \ea \right) 
    \end{align}
   from which we obtain using the structure of $A$ in (\ref{AStruct}):
    \be
    \lim_{i\to\infty} \bmu_{\mathcal{R},i}(\theta)=T_{SR}\tran\lim_{i\to \infty}\bmu_{\mathcal{S},i}(\theta)+T_{RR}\tran\lim_{i\to\infty}\bmu_{\mathcal{R},i}(\theta)
    \ee
    We then conclude that
    \begin{align} 
    \lim_{i\to\infty}\bmu_{\mathcal{R},i}(\theta)&=(I-T_{RR}\tran)^{-1}T_{SR}\tran\left(\lim_{i\to\infty}\bmu_{\mathcal{S},i}(\theta)\right) \nn \\
    &=W\tran\left(\lim_{i\to\infty}\bmu_{\mathcal{S},i}(\theta)\right)
    \end{align}
    \qd
   
     \section{Proof of Lemma \ref{lemma2}} \label{app.A}
     We start by introducing some notation and definitions. Since we are now interested in examining the evolution of the agents' beliefs toward the true state, let us introduce the true probability mass function $p(\theta)$ defined over $\Theta$, namely: 
     \be
     p(\theta)=
     \delta_{\theta,\theta^\circ}
     \define \left\{
     \begin{aligned}
     	1,&\quad {\rm if}\quad \theta = \theta^\circ\\
     	0,&\quad {\rm otherwise}
     \end{aligned}
     \right. 
     \label{eqn:trueProba1}
     \ee
     The evolution of the belief of agent $k$  toward the true state can be analyzed by computing the KL divergence of $\bm\mu_{k,i}(\theta)$ from $p(\theta)$ at each time instant $i$. We therefore introduce the new regret function for agent $k$ at time $i$ as: 
     \bq
     Q(\bm\mu_{k,i})&\define& D_{KL}(p||\bm\mu_{k,i})=\sum_{\theta\in \Theta} p(\theta)\log\left(\frac{p(\theta)}{\bm\mu_{k,i}(\theta)}\right)\nn \\ 
     &=& -\log \bm\mu_{k,i}(\theta^\circ)
     \label{riskStrongly}
     \eq
     where we used the convention that $0\log0=0$. We shall again define agent $k$'s individual risk at time $i$ as
     \be
     J(\bmu_{k,i})\define \Ex_{\mathcal{F}_i} Q(\bmu_{k,i})=-\Ex_{\mathcal{F}_i}\log \bmu_{k,i}(\theta^\circ)
     \ee
     where $\mathcal{F}_i$ denotes the history of $\{\bm{\xi}_{k,j}\}$ over all $k$ and for all $j\leq i$. We then assess the overall network performance by considering the weighted aggregate risk:
     \be
     J(\bmu_i)\define \sum_{k=1}^{N}y(k)J(\bmu_{k,i})
     \ee
     where the $\{y(k)\}$ denote the entries of the Perron vector, $y$, of the primitive left-stochastic matrix $A$, as defined by (\ref{perronFrob}). To prove Lemma \ref{lemma2}, namely, the ability of agents to arrive at correct forecasts, we prove first the convergence of the sequence $\{J(\bmu_i)\}$ as $i\to\infty$. This convergence will then imply the correct forecasting by agents. 
     
     \vspace{0.5cm}
     \noindent {\bf Proof of Lemma \ref{lemma2}}:  We assumed in the statement of the lemma that at least one agent $\ell_o$ starts with a non-zero prior belief at $\theta^\circ$, i.e., $\bmu_{\ell_o,0}(\theta^\circ)>0$. As shown in \cite{zhao2012learning}, this condition guarantees that for large enough $i$, $\bmu_{k,i}(\theta^\circ)>0$ for all $k \in \mathcal{N}$, which implies that the terms of the time sequence $\{Q(\bmu_{k,i})\}$ assume nonnegative values for large $i$ and for any agent $k$. Thus, the time sequences $\{J(\bmu_{k,i})\}$ and $\{J(\bmu_{i})\}$ are non-negative for large enough $i$. Let us now expand agent $k$'s risk for large time $i$:
     \begin{align}
     	J(\bm\mu_{k,i})&=-\Ex_{\mathcal{F}_i} \log \bm{\mu}_{k,i}(\theta^\circ) \nn \\
     	{} &=-\Ex_{\mathcal{F}_i} \log \left(\sum_{\ell \in \mathcal{N}_k } a_{\ell k}\bm{ \psi}_{\ell,i} (\theta^\circ)\right)\nn \\
     	{}&\stackrel{(a)}\leq-\Ex_{\mathcal{F}_i}\left[\sum_{\ell \in \mathcal{N}_k } a_{\ell k}\log\left(\bm\psi_{\ell,i}(\theta^\circ)\right)\right]\nn \\ 
     	{}&\stackrel{(\ref{model2})}=-\Ex_{\mathcal{F}_i}\left[\sum_{\ell \in \mathcal{N}_k } a_{\ell k}\log\biggl(\left(1-\gamma_{\ell,i}\right)\left(\bm{\mu}_{\ell,i-1}(\theta^\circ)\right)\right.\nn \\
     	&\qquad+\left.\gamma_{\ell,i}\left(\frac{ \bm{\mu}_{\ell,i-1} (\theta^\circ) L_\ell(\bm\xi_{\ell,i} | \theta^\circ) } {\bm m_{\ell,i-1}(\bm\xi_{\ell,i})} \right)\biggr)\right]\nn \\
     	{}&\stackrel{(b)}\leq-\Ex_{\mathcal{F}_i}\left[\sum_{\ell \in \mathcal{N}_k } a_{\ell k}\biggl(\left(1-\gamma_{\ell,i}\right)\log\left(\bm{\mu}_{\ell,i-1}(\theta^\circ)\right) \right. \nn \\
     	&\qquad+\left.\gamma_{\ell,i}\log\left(\frac{ \bm{\mu}_{\ell,i-1} (\theta^\circ) L_\ell(\bm\xi_{\ell,i} | \theta^\circ) } {\bm m_{\ell,i-1}(\bm\xi_{\ell,i})} \right)\biggr)\right]\nn \\
     	{}&=-\Ex_{\mathcal{F}_i}\left(\sum_{\ell \in \mathcal{N}_k } a_{\ell k}\log\left(\bm{\mu}_{\ell,i-1}(\theta^\circ)\right) \right) \nn \\
     	&\qquad-\Ex_{\mathcal{F}_i}\left(\sum_{\ell \in \mathcal{N}_k } a_{\ell k}\gamma_{\ell,i}\log\left(\frac{ L_\ell(\bm\xi_{\ell,i} | \theta^\circ) } {\bm m_{\ell,i-1}(\bm\xi_{\ell,i})} \right)\right)\nn \\
     	{}&\stackrel{(c)}= -\Ex_{\mathcal{F}_i}\left(\sum_{\ell \in \mathcal{N}_k} a_{\ell k} \log\left(\bm{\mu}_{\ell,i-1}(\theta^\circ)\right)\right)-\Ex_{\mathcal{F}_{i-1}} \nn \\
     	{}&\quad\left(\sum_{\ell \in \mathcal{N}_k} a_{\ell k}\gamma_{\ell,i}\Ex_{\bm\xi_{\ell,i}}\left[\log\left(\frac{ L_\ell(\bm\xi_{\ell,i} | \theta^\circ) } {\bm m_{\ell,i-1}(\bm\xi_{\ell,i})}\right)|\mathcal{F}_{i-1}\right]\right)\nn \\
     	{}&\stackrel{(d)}\leq -\sum_{\ell \in \mathcal{N}_k } a_{\ell k}\Ex_{\mathcal{F}_i}\log\left(\bm{\mu}_{\ell,i-1}(\theta^\circ)\right) \nn \\
     	&= \sum_{\ell \in \mathcal{N}_k}a_{\ell k}J(\bmu_{\ell,i-1}) \label{proof1}
     \end{align}
     where 
     \begin{itemize}
     	\item steps $(a)$ and $(b)$ follow from the convexity of $-\log(.)$;
     	\item step $(c)$ follows from the conditional expectation property $\left(\Ex_X[g(X)]=\Ex_Y[\Ex_{X|Y}[g(X)|Y]]\right)$  as in (\ref{condProperty});
     	\item step $(d)$ follows by replacing the expression in $(c)$ by an upper bound using the non-negativity of the KL divergence from $L_\ell(.|\theta^\circ)$ to $\bm\m_{\ell,i-1}(.)$ according to Gibb's inequality \cite{Cover}.
     \end{itemize}
     Accordingly, the overall performance at time $i$, satisfies:
     \begin{align}
     	J(\bm\mu_{i})&\stackrel{(a)}\leq\sum_{k=1}^{N}y(k)\sum_{\ell \in \mathcal{N}_k } a_{\ell k}J(\bm\mu_{\ell,i-1}) \nn\\
     	{}&\stackrel{(b)}=\sum_{\ell=1}^{N}y({\ell})J(\bm\mu_{\ell,i-1})= J(\bm\mu_{i-1})    
     \end{align}
     where step $(a)$ follows from (\ref{proof1}), and step $(b)$ follows from (\ref{perronFrob}). Therefore, the sequence $\{J(\bm\mu_{i})\}$ is a decreasing sequence. But, since this sequence is non-negative, we conclude that $\{J(\bmu_i)\}$ converges to a real number according to the monotone convergence theorem of real numbers \cite{rudin}.
     
     We now establish the ability of agents to attain correct predictions.
     From step $(c)$ in (\ref{proof1}), we get
     \begin{align}
     	J(\bm\mu_{k,i})	\leq \sum_{\ell \in \mathcal{N}_k} a_{\ell k}J(\bm\mu_{\ell,i-1}) \qquad \qquad \qquad \qquad \qquad \qquad \quad  &\nn \\
     	-\Ex_{\mathcal{F}_{i-1}}\left(\sum_{\ell \in \mathcal{N}_k} a_{\ell k}\gamma_{\ell,i}\Ex_{\bm\xi_{\ell,i}}\left[\log\left(\frac{ L_\ell(\bm\xi_{\ell,i} | \theta^\circ) } {\bm m_{\ell,i-1}(\bm\xi_{\ell,i})}\right)|\mathcal{F}_{i-1}\right]\right) \nn \\
     \end{align}
     Then, rearranging terms,
     \begin{align}
     	\sum_{\ell \in \mathcal{N}_k } a_{\ell k}J(\bm\mu_{\ell,i-1})-J(\bm\mu_{k,i})\geq\qquad \qquad \qquad \qquad \qquad \qquad   &\nn \\
     	\Ex_{\mathcal{F}_{i-1}}\left(\sum_{\ell \in \mathcal{N}_k} a_{\ell k}\gamma_{\ell,i}\Ex_{\bm\xi_{\ell,i}}\left[\log\left(\frac{ L_\ell(\bm\xi_{\ell,i} | \theta^\circ) } {\bm m_{\ell,i-1}(\bm\xi_{\ell,i})}\right)|\mathcal{F}_{i-1}\right]\right)\nn  \\
     \end{align}
     Scaling by $y(k)$, summing over $k$, and using (\ref{perronFrob}) we get:
     \begin{align}
     	\sum_{\ell=1}^{N}y(\ell) J(\bm\mu_{\ell,i-1})-\sum_{k=1}^{N}y(k)J(\bm\mu_{k,i}) \geq\qquad \qquad \qquad  &\nn\\
     	\Ex_{\mathcal{F}_{i-1}}\left(\sum_{\ell \in \mathcal{N}_k} y(\ell)\gamma_{\ell,i}\Ex_{\bm\xi_{\ell,i}}\left[\log\left(\frac{ L_\ell(\bm\xi_{\ell,i} | \theta^\circ) } {\bm m_{\ell,i-1}(\bm\xi_{\ell,i})}\right)|\mathcal{F}_{i-1}\right]\right) 
     \end{align}
     Then,
     \begin{align}
     	J(\bm\mu_{i-1})-J(\bm\mu_{i}) \geq\qquad \qquad \qquad \qquad \qquad \qquad \qquad  &\nn \\ \Ex_{\mathcal{F}_{i-1}}\left(\sum_{\ell \in \mathcal{N}_k} y(\ell)\gamma_{\ell,i}\Ex_{\bm\xi_{\ell,i}}\left[\log\left(\frac{ L_\ell(\bm\xi_{\ell,i} | \theta^\circ) } {\bm m_{\ell,i-1}(\bm\xi_{\ell,i})}\right)|\mathcal{F}_{i-1}\right]\right)
     \end{align}
     Since $\{J(\bmu_{i})\}$ is a convergent sequence, it is also a Cauchy sequence \cite{rudin} and, therefore,
     \begin{align}
     	0=\lim_{i\to\infty}\left[J(\bm\mu_{i-1})-J(\bm\mu_{i})\right]\geq\lim_ {i\to\infty} \qquad \qquad \qquad \qquad  &\nn \\ 
     	\Ex_{\mathcal{F}_{i-1}}\left(\sum_{\ell \in \mathcal{N}_k} y(\ell)\gamma_{\ell,i}\Ex_{\bm\xi_{\ell,i}}\left[\log\left(\frac{ L_\ell(\bm\xi_{\ell,i} | \theta^\circ) } {\bm m_{\ell,i-1}(\bm\xi_{\ell,i})}\right)|\mathcal{F}_{i-1}\right]\right) \nn & \\ 
     	\geq 0 \qquad \quad
     \end{align}
     where the rightmost inequality follows from the non-negativity of the KL-divergence. We conclude that: 
     \begin{align}
     	\lim_ {i\to\infty}\Ex_{\mathcal{F}_{i-1}}\left(\sum_{\ell \in \mathcal{N}_k} y(\ell)\gamma_{\ell,i}\right.\quad\quad\quad\quad\quad\quad\quad\quad\quad\quad\quad&\nn \\
     \left.\Ex_{\bm\xi_{\ell,i}}\left[\log\left(\frac{ L_\ell(\bm\xi_{\ell,i} | \theta^\circ) } {\bm m_{\ell,i-1}(\bm\xi_{\ell,i})}\right)|\mathcal{F}_{i-1}\right]\right)=0  &
     \end{align}
     Since we assumed that $\lim \limits_{i \to \infty} \gamma_{k,i}\neq0 $ for any $k$, $y(\ell)>0$  from (\ref{perronFrob}), and $\Ex_{\bm\xi_{\ell,i}}\left[\log\left(\frac{ L_\ell(\bm\xi_{\ell,i} | \theta^\circ) } {\bm m_{\ell,i-1}(\bm\xi_{\ell,i})}\right)|\mathcal{F}_{i-1}\right]\geq 0$ from the non-negativity of the KL-divergence, then
     \be
     \lim_ {i\to\infty}\Ex_{\bm\xi_{\ell,i}}\left[\log\left(\frac{ L_\ell(\bm\xi_{\ell,i} | \theta^\circ) } {\bm m_{\ell,i-1}(\bm\xi_{\ell,i})}\right)|\mathcal{F}_{i-1}\right]=0
     \ee
     Thus,
     \be
     \lim_{i\to\infty}\sum_{\zeta_\ell\in Z_\ell}L_\ell(\zeta_\ell|\theta^\circ)\log\left( \frac{L_\ell(\zeta_\ell|\theta^\circ)}{m_{\ell,i-1}(\zeta_\ell)}\right)=0
     \label{KLeq}
     \ee
     Let
     \be
     f_{\ell,i-1}\define\sum_{\zeta_\ell\in Z_\ell}L_\ell(\zeta_\ell|\theta^\circ)\log\left( \frac{L_\ell(\zeta_\ell|\theta^\circ)}{m_{\ell,i-1}(\zeta_\ell)}\right)
     \label{KLeqb}
     \ee
     where $f_{\ell,i}$ represents the KL-divergence of $m_{\ell,i}(.)$ from $L_{\ell}(.|\theta^\circ)$. We know from Gibb's inequality \cite{Cover} that the KL-divergence of a probability distribution from another distribution achieves the value zero only when the two distributions are equal. Since the KL-divergence $f_{\ell,i}$ converges to zero as $i\to\infty$ and $L_{\ell}(.|\theta^\circ)$ is a fixed distribution, this implies that $m_{\ell,i}(.)$ should converge, i.e., its limit exists and it takes the following value:
     \be
     \lim_{i \to\infty} m_{\ell,i}(\zeta_{\ell})=L_\ell(\zeta_{\ell}|\theta^\circ)
     \ee
     for any $\zeta_\ell \in Z_\ell$. Since this result is achieved for any realization of observational signals $\mathcal{F}_{i-1}$, we conclude that:
     \be
     \lim_{i \to\infty} \bm m_{\ell,i}(\zeta_{\ell})\aseq L_\ell(\zeta_{\ell}|\theta^\circ)
     \ee
     for any $\ell \in \mathcal{N}$ and any $\zeta_\ell \in \Z_\ell$.
     \qd 
    \section{Proof of Theorem \ref{mainresultself}} \label{App.C}
    According to model (\ref{model2}), the intermediate belief of any agent $k$ in a receiving group can be written as follows:
    \begin{align}
    &\bm\psi_{k,i}(\theta) = \bm\mu_{k,i-1}(\theta) \nn \\
    &+\gamma_{k,i}\left[\bm\mu_{k,i-1}(\theta)\left(\frac{  L_k(\bm\xi_{k,i} | \theta) } {\sum_{\theta '} \bm\mu_{k,i-1} (\theta ') L_k(\bm\xi_{k,i} | \theta')}-1\right)\right] 
    \end{align}
    We assume that $\gamma_{k,i}=\tau_{k,i}\gamma_{\max}$, where $\tau_{k,i}$ and $\gamma_{\max}$ are both nonnegative scalars less than one. Then,
    \begin{align}
    &\bm\psi_{k,i}(\theta) =\bm\mu_{k,i-1}(\theta) \nn \\
    &+\gamma_{\max}\left[\tau_{k,i}\bm\mu_{k,i-1}(\theta)\left(\frac{  L_k(\bm\xi_{k,i} | \theta) } {\sum_{\theta ' } \bm\mu_{k,i-1} (\theta ') L_k(\bm\xi_{k,i} | \theta')}-1\right)\right] 
    \end{align}
    We define the auxiliary function:
    \begin{align}
    \hspace*{-0.15cm}\bm h_{k,i}(\theta,\zeta_k) \define
    \tau_{k,i}\bm\mu_{k,i-1}(\theta)\left(\frac{  L_k(\zeta_k | \theta) } {\sum_{\theta '} \bm\mu_{k,i-1} (\theta ') L_k(\zeta_{k} | \theta')}-1\right)
    \label{hdefinition}
    \end{align}
    where $\theta \in \Theta$ and $\zeta_k\in Z_k$, so that
    \begin{align}
    \bm\psi_{k,i}(\theta) = \bm\mu_{k,i-1}(\theta)+\gamma_{\max} \bm h_{k,i} (\theta,\bm\xi_{k,i})
    \end{align}
   Therefore,
    \begin{align}
    \bm\mu_{k,i}(\theta) &=\sum_{\ell \in \mathcal{N}_k} a_{\ell k} \bm \psi_{\ell,i}(\theta) \nn\\
    &= \sum_{\ell \in \mathcal{N}_k } a_{\ell k} \bm\mu_{\ell,i-1} (\theta) + \gamma_{\max}\sum_{\ell \in \mathcal{N}_k } a_{\ell k} \bm h_{\ell,i}(\theta,\bm\xi_{\ell,i})
    \end{align}
    
   \noindent Let us introduce the vectors:
     \begin{align}
     	\hspace{-0.1cm}\bm h_{\mathcal{S},i}(\theta,\bm\xi_{S,i})\define&\;
     	{\rm col}\left\{
     	\bm h_{1,i}(\theta,\bm\xi_{1,i}),
     	\hdots,
     	\bm h_{N_{gS},i}(\theta,\bm\xi_{N_{gS},i}) 
     	\right\} 
     	\end{align}
     	\vspace{-0.5cm}
        \begin{align}
     	\hspace{-0.13cm}\bm h_{\mathcal{R},i}(\theta,\bm\xi_{R,i})\define\quad\quad\quad\quad\quad\quad\quad\quad\quad\quad\quad\quad\quad\qquad&\nn \\
     	{\rm col} \left\{
        \bm h_{N_{gS}+1,i}(\theta,\bm\xi_{N_{gS}+1,i}),
     	\hdots,
        \bm h_{N,i}(\theta,\bm\xi_{N,i})
     	\right\}&
     	\end{align}
     	and,
        \begin{align}
     	\hspace{-0.1cm} \bm\xi_{\mathcal{S},i}\define&\;
     	{\rm col}\left\{
     	\bm\xi_{1,i},
     	\hdots,
     	\bm\xi_{N_{gS},i} 
     	\right\} \\
     	\hspace{-0.1cm} \bm\xi_{\mathcal{R},i}\define&\;
     	{\rm col} \left\{
     	\bm\xi_{N_{gS}+1,i},
     	\hdots,
     	\bm\xi_{N,i}
     	\right\}
     \end{align}
      
    Recall that we defined the vectors $\bmu_{\mathcal{S},i}(\theta)$ in (\ref{muS}) and $\bmu_{\mathcal{R},i}(\theta)$ in (\ref{muR}). Then, we have

    \begin{align}
    \hspace{-0.3cm}
    \ba {c}
    \hspace{-0.2cm}\bmu_{\mathcal{S},i} (\theta)\hspace{-0.2cm} \\
    \hspace{-0.2cm}\bmu_{\mathcal{R},i}(\theta) \hspace{-0.2cm}
    \ea
    &= A\tran
    \left(\hspace{-0.05cm}\ba {c}\hspace{-0.2cm}
    \bmu_{\mathcal{S},i-1} (\theta) \hspace{-0.2cm}\\\hspace{-0.2cm}
    \bmu_{\mathcal{R},i-1} (\theta) \hspace{-0.2cm}
    \ea +\gamma_{\max}\ba {c}
   \hspace{-0.2cm} \bm h_{\mathcal{S},i} (\theta,\bm\xi_{\mathcal{S},i}) \hspace{-0.2cm}\\
    \hspace{-0.2cm}\bm h_{\mathcal{R},i} (\theta,\bm\xi_{\mathcal{R},i}) \hspace{-0.2cm}
    \ea \hspace{-0.05cm}\right)   
    \end{align}

   \noindent Using the structure of $A$ in (\ref{AStruct}), it follows that
    \begin{align}
    \bmu_{\mathcal{R},i}(\theta)&= T_{RR}\tran\bmu_{\mathcal{R},i-1} (\theta)+ T_{SR}\tran\bmu_{\mathcal{S},i-1}(\theta)  \nn  \\
    &\;\;+ \gamma_{\max} \left( T_{SR}\tran \bm h_{\mathcal{S},i}(\theta,\bm\xi_{\mathcal{S},i})+T_{RR}\tran \bm h_{\mathcal{R},i}(\theta,\bm\xi_{\mathcal{R},i}) \right) \label{Ineq}
    \end{align}
    We study the convergence of this recursion. Let 
     \begin{align}
     	\bm\zeta_{\mathcal{S}}\define
     	{\rm col}\left\{
        \bm\zeta_{1},
     	\hdots,
        \bm\zeta_{N_{gS}} 
     	\right\},\;\bm\zeta_{\mathcal{R}}\define
     	{\rm col} \left\{
     	 \bm\zeta_{N_{gS}+1},
     	\hdots,
     	 \bm\zeta_{N} 
     	\right\}
     \end{align}
    We will first establish that 
    \be
    \gamma_{\max} \left( T_{SR}\tran \bm h_{\mathcal{S},i}(\theta,\zeta_{\mathcal{S}})+T_{RR}\tran \bm h_{\mathcal{R},i}(\theta,\zeta_{\mathcal{R}})\right)=O(\gamma_{\max}) \label{bigOresult}
    \ee
    for any $\theta$, $\zeta_\mathcal{S}$ and $\zeta_\mathcal{R}$.    
    \begin{lemma}
    	For any $k \in \mathcal{N}$, $i\geq0$, $\theta \in \Theta$ and $\zeta_k \in Z_k$, it holds that
    	\be
    	|\bm h_{k,i}(\theta,\zeta_k)| \leq 1
    	\label{lemmaBound}
    	\ee 
    \end{lemma}
    \begin{proof}
    From (\ref{hdefinition}),
    \be
    \bm h_{k,i}(\theta,\zeta_k) = \tau_{k,i}\left(\frac{ \bm\mu_{k,i-1}(\theta) L_k(\zeta_k | \theta) } {\sum_{\theta '} \bm\mu_{k,i-1} (\theta ') L_k(\zeta_{k} | \theta')}-\bm\mu_{k,i-1}(\theta)\right)
    \ee
    Since $\tau_{k,i}$ is a nonnegative scalar that is less than one, and since
    \be 
    0\leq\frac{ \bm\mu_{k,i-1}(\theta) L_k(\zeta_k | \theta) } {\sum_{\theta '} \bm\mu_{k,i-1} (\theta ') L_k(\zeta_{k} | \theta')} \leq1
    \ee
    for any $k \in \mathcal{N}$, $i\geq0$, $\theta \in \Theta$ and $\zeta_k \in Z_k$, we conclude that
    \be
    \bm h_{k,i}(\theta,\zeta_k)\geq-\tau_{k,i}\bmu_{k,i-1}(\theta)\geq-\bmu_{k,i-1}(\theta)
    \ee
    and
    \be
    \bm h_{k,i}(\theta,\zeta_k)\leq \tau_{k,i}\left(1-\bmu_{k,i-1}(\theta)\right)\leq 1-\bmu_{k,i-1}(\theta)
    \ee
    Moreover, we know that 
       $ 0\leq\bmu_{k,i}(\theta)\leq 1$ for all $k$, $i$ and $\theta$. We then conclude that
    \be
    -1\leq\bm h_{k,i}(\theta,\zeta_k)\leq1 
    \ee 
    \end{proof}
  \noindent From (\ref{lemmaBound}), we get for any $i\geq0$ and $\theta\in\Theta$,
    \bq
    \hspace*{-3cm} \left|T_{SR}\tran \bm h_{\mathcal{S},i}(\theta,\zeta_{\mathcal{S}})+T_{RR}\tran \bm h_{\mathcal{R},i}(\theta,\zeta_{\mathcal{R}})\right| \nn
     \eq
     \vspace*{-0.5cm}
     \begin{align}
     &\preceq T_{SR}\tran \bm |\bm h_{\mathcal{S},i}(\theta,\zeta_{\mathcal{S}})|+T_{RR}\tran \bm |\bm h_{\mathcal{R},i}(\theta,\zeta_{\mathcal{R}})| \nn \\ 
     &\stackrel{(a)}\preceq T_{SR}\tran  \one_{N_{gS}} +T_{RR}\tran  \one_{N_{gR}}\stackrel{(b)}= \one_{N_{gR}}
    \end{align}
     where $(a)$ follows from (\ref{lemmaBound}) and $(b)$ follows from the left-stochasticity of the combination matrix $A$. Note that the above inequality, as well as the absolute value operator, are element-wise. Moreover, $\one_{N_{gR}}$ is a vector of all ones of size $N_{gR}$ and $\one_{N_{gS}}$ is a vector of all ones of size $N_{gS}$.
    Thus,
    \begin{align} 
     \hspace*{-0.3cm}\gamma_{\max}\left|T_{SR}\tran \bm h_{\mathcal{S},i}(\theta,\zeta_{\mathcal{S}})+T_{RR}\tran \bm h_{\mathcal{R},i}(\theta,\zeta_{\mathcal{R}})\right| \preceq \gamma_{\max} \one_{N_{gR}} \label{finally}
    \end{align}
    for all $i\geq0$. This fact leads to the desired conclusion (\ref{bigOresult}). In this way, equality (\ref{Ineq}) implies:
    \begin{align}
    \bmu_{\mathcal{R},i}(\theta) &\preceq T_{RR}\tran\bmu_{\mathcal{R},i-1} (\theta)+ T_{SR}\tran\bmu_{\mathcal{S},i-1}(\theta)+ \gamma_{\max}\one_{N_{gR}} \nn \\
    \bmu_{\mathcal{R},i}(\theta) &\succeq T_{RR}\tran\bmu_{\mathcal{R},i-1} (\theta)+ T_{SR}\tran\bmu_{\mathcal{S},i-1}(\theta)- \gamma_{\max}\one_{N_{gR}}\nn \\
    \label{Ineq2}
    \end{align}
   We have $\rho(T_{RR}\tran)<1$ and $\lim_{i\to\infty}\bmu_{\mathcal{S},i}(\theta)$ exists since agents of sending sub-networks can learn asymptotically the truth. Then,
   \begin{align}
    \hspace*{-0.3cm}\limsup_{i\to\infty}\bmu_{\mathcal{R},i}(\theta) &\preceq T_{RR}\tran\left(\limsup_{i\to\infty}\bmu_{\mathcal{R},i-1} (\theta)\right)\nn \\
   &{}+ T_{SR}\tran\left(\lim_{i\to\infty}\bmu_{\mathcal{S},i-1}(\theta)\right)+ \gamma_{\max}\one_{N_{gR}} \nn \\
   \hspace*{-0.5cm}\liminf_{i\to\infty}\bmu_{\mathcal{R},i}(\theta) &\succeq T_{RR}\tran\left(\liminf_{i\to\infty}\bmu_{\mathcal{R},i-1} (\theta)\right) \nn \\
   &{}+T_{SR}\tran\left(\lim_{i\to\infty}\bmu_{\mathcal{S},i-1}(\theta)\right)- \gamma_{\max}\one_{N_{gR}}  \label{Ineq2Limits} 
   \end{align}
   It follows that
   \begin{align}
     \limsup_{i\to\infty}\bmu_{\mathcal{R},i}(\theta) &\preceq W\tran\left(\lim_{i\to\infty}\bmu_{\mathcal{S},i}(\theta)\right)+ \gamma_{\max}C\one_{N_{gR}} \nn \\
     \liminf_{i\to\infty}\bmu_{\mathcal{R},i}(\theta) &\succeq  W\tran\left(\lim_{i\to\infty}\bmu_{\mathcal{S},i}(\theta)\right)- \gamma_{\max}C\one_{N_{gR}}  \label{Ineq3Limits2}
     \end{align}
     where $C=(I-T_{RR}\tran)^{-1}$. \qd
    \bibliographystyle{IEEEbib}

    \newpage
    \begin{IEEEbiographynophoto}{Hawraa Salami}
    	(S'16) received her B.S. and M.S. degrees from American University of Beirut (AUB) and University of California, Los Angeles (UCLA) in 2013 and 2014, respectively. She is currently working towards the PhD degree in Electrical Engineering at UCLA. Her research interests include social learning, social network modelling, multi-agent network processing and statistical signal processing.
    \end{IEEEbiographynophoto}
    \vspace{-5in}
    \begin{IEEEbiographynophoto}{Bicheng Ying}
    	(S'15) received his B.S. and M.S. degrees from Shanghai Jiao Tong University (SJTU) and University of California, Los Angeles (UCLA) in 2013 and 2014, respectively. He is currently working towards the PhD degree in Electrical Engineering at UCLA. His research interests include multi-agent network processing, large-scale machine learning, distributed optimization, and statistical signal processing.
    \end{IEEEbiographynophoto}
      \vspace{-5in} 
    \begin{IEEEbiographynophoto}{Ali H. Sayed}
    	 (S'90-M'92-SM'99-F'01) is distinguished professor and former chairman of electrical engineering at the University of California, Los Angeles, where he directs the UCLA Adaptive Systems Laboratory. An author of over 480 scholarly publications and six books, his research involves several areas including adaptation and learning, statistical signal processing, distributed processing, network science, and biologically-inspired designs. His work has been recognized with several awards including the 2014 Athanasios Papoulis Award from the European Association for Signal Processing, the 2015 Education Award, the 2013 Meritorious Service Award, and the 2012 Technical Achievement Award from the IEEE Signal Processing Society, the 2005 Terman Award from the American Society for Engineering Education, the 2003 Kuwait Prize, and the 1996 IEEE Donald G. Fink Prize. He served as Distinguished Lecturer for the IEEE Signal Processing Society in 2005 and as Editor-in Chief of the IEEE TRANSACTIONS ON SIGNAL PROCESSING (2003-2005). His articles received several Best Paper Awards from the IEEE Signal Processing Society in 2002, 2005, 2012, and 2014. He is a Fellow of both the IEEE and the American Association for the Advancement of Science (AAAS). He is recognized as a Highly Cited Researcher by Thomson Reuters.
    \end{IEEEbiographynophoto}
 \end{document}